\documentclass[12pt, draftclsnofoot, onecolumn]{IEEEtran}
\usepackage[colorlinks,
linkcolor=red,
anchorcolor=green,
citecolor=blue
]{hyperref} 
\usepackage{amsfonts,amssymb,amsmath,array}
\usepackage{latexsym}
\usepackage{CJK}
\usepackage{cite}
\usepackage{bm}
\usepackage{color, soul}
\usepackage{colortbl}

\usepackage{graphicx}
\usepackage{epstopdf}
\usepackage{epsfig}
\usepackage{subfigure} 
\usepackage{framed}
\usepackage{verbatim}  
\usepackage{color,soul}
\definecolor{aliceblue}{rgb}{0.94, 0.97, 1.0}
\definecolor{blizzardblue}{rgb}{0.67, 0.9, 0.93}
\definecolor{antiquebrass}{rgb}{0.8, 0.58, 0.46}
\definecolor{beaublue}{rgb}{0.74, 0.83, 0.9}
\sethlcolor{blizzardblue}
\usepackage{mathrsfs}
\usepackage{algorithm} 
\usepackage{algorithmic} 
\usepackage{booktabs}
\usepackage{textcomp}
\usepackage{multirow}
\usepackage{lettrine}
\usepackage[scaled=0.92]{helvet}

\usepackage{algorithm}
\usepackage{algorithmic}
\usepackage{subfigure} 
\usepackage{framed}
\usepackage{verbatim}  

\usepackage{diagbox} 
\usepackage{tcolorbox}

\usepackage{amsthm} 
\theoremstyle{remark} 

\newtheorem{thm}{Theorem}
\newtheorem{lem}{Lemma}
\newtheorem{prop}{Proposition}
\newtheorem{rem}{Remark}

\usepackage{accents} 
\makeatletter
\def\widebar{\accentset{{\cc@style\underline{\mskip10mu}}}}
\def\Widebar{\accentset{{\cc@style\underline{\mskip13mu}}}}
\begin{document}
\title{Matching Users’ Preference Under Target Revenue Constraints in Optimal Data Recommendation Systems}

\author{Shanyun~Liu, Yunquan Dong, Pingyi~Fan,~\IEEEmembership{Senior Member, ~IEEE}, Rui~She, Shuo~Wan\\\

\thanks{
Shanyun~Liu, Rui~She, Shuo~Wan and Pingyi Fan are with Tsinghua National Laboratory for Information Science and Technology(TNList) and the Department of Electronic Engineering,  Tsinghua University,  Beijing,  P.R. China, 100084. e-mail: \{liushany16,~sher15,~wan-s17@mails.tsinghua.edu.cn, fpy@tsinghua.edu.cn.\} }
\thanks{Yunquan Dong is with the School of Electronic and Information Engineering, Nanjing University of Information Science and Technology, Nanjing, P.R. China. email: yunquandong@nuist.edu.cn
}
\thanks{This work was supported in part by the National Natural Science Foundation of China (NSFC) under Grant 61771283, and in part by the China Major State Basic Research Development Program (973 Program) under Grant 2012CB316100(2).}
}

\maketitle

\graphicspath{{Figures/}}
 \vspace{-25mm}
\begin{abstract}
This paper focuses on the problem of finding a particular data recommendation strategy based on the user preferences and a system expected revenue. To this end, we formulate this problem as an optimization by designing the recommendation mechanism as close to the user behavior as possible with a certain revenue constraint. In fact, the optimal recommendation distribution is the one that is the closest to the utility distribution in the sense of relative entropy and satisfies expected revenue. We show that the optimal recommendation distribution follows the same form as the message importance measure (MIM) if the target revenue is reasonable, i.e., neither too small nor too large. Therefore, the optimal recommendation distribution can be regarded as the normalized MIM, where the parameter, called importance coefficient, presents the concern of the system and switches the attention of the system over data sets with different occurring probability. By adjusting the importance coefficient, our MIM based framework of data recommendation can then be applied to system with various system requirements and data distributions.
 Therefore, the obtained results illustrate the physical meaning of MIM from the data recommendation perspective and validate the rationality of MIM in one aspect.
\end{abstract}


\IEEEpeerreviewmaketitle
\section{Introduction}
The data recommendation is one of the most fundamental problems in wireless mobile Internet, and it becomes more and more crucial in the era of big data.
With the explosive growth of data, it is difficult to send all the data to the users within tolerable time by traditional data processing technology \cite{chen2014big,bi2015wireless}.
Data push or recommendation technology may help users get their desired information in time. In fact, data can be sent in advance when the system is idle, even if there is no ask for it, rather than wasting time to wait for a clear request from users \cite{franklin1998data}. Furthermore, it is also arduous for users to find desired data among a mass of data available in the Internet \cite{hauswirth1999internet}. In general, the special data which catches the interests of users can be provided faster and easier with data recommendation system, since the recommendation strategy is usually well-designed based on the preferences of users explicitly \cite{tadrous2015proactive,shoukry2014proactive}. Compared to search engines, push is more convenient for less action required, and the quality of data pushed does not rely on the skills and knowledge of the users \cite{hauswirth1999internet}. For example, many applications in mobile phone prefer to recommend some data based on the user's interest to improve the user experience. In addition, different from traditional industry, the Internet enterprises can make a profit through mobile network by using push-based advertisement technology \cite{kim2009mobile}.

The previous works mainly discussed data push based on the content to solve the problem of data delivery \cite{franklin1998data,hauswirth1999internet,tadrous2015proactive,shoukry2014proactive,podnar2002mobile,kim2009mobile,nicopolitidis2005an,bhide2002adaptive}. Data push is usually regarded as a strategy of data delivery in distributed information systems \cite{franklin1998data}. The architecture for a mobile push system was proposed in \cite{hauswirth1999internet,podnar2002mobile}. In addition, \cite{nicopolitidis2005an} put forward an effective wireless push system for high-speed data broadcasting. The push and pull techniques for time-varying data network was discussed in \cite{bhide2002adaptive}. Furthermore, on this basis, the recommendation system was provided as information-filtering system which push data to users based on the knowledge about their preferences \cite{li2018joint,parra2014optimal}. \cite{li2018joint} discussed the joint content recommendation, and the privacy-enhancing technology in recommendation systems was investigated in \cite{parra2014optimal}. Besides, \cite{zhang2017personalized} put forward a personalized social image recommendation method. Furthremore, the recommendation technology was also used to solve the problem in multimedia big data \cite{zhou2016differentially}. Different from the preceding works, in this paper, we do not discuss the problem of data delivery based on content. As an alternative, we shall discuss the distribution of recommendation sequence based on the performance of users when the revenue of the recommendation system is required. 

Nowadays, user customization is crucial and it can be solved by means of recommendation. The personalized concern of user usually can be characterized by many properties of the data used by users, such as data format, keywords \cite{hauswirth1999internet,verbert2012context,cheng2016on}. However, we choose the frequency of data used to describe the performance of users since it has nothing to do with the concrete content. 

In addition, we take the revenue of the recommendation system into account. That is, a recommendation process consumes resources while gets benefits. Different recommendation types may bring different results. For example, the cost of omitting a data which should be pushed mistakenly may be much smaller than that of pushing invalid data to a user erroneously. On one hand, the users have confidence in the recommendation system if the desired data is pushed to these users correctly. On the other hand, pushing the information that users do not ask, such as advertisement, may seriously impact the user experience, but it can still bring revenue from advertisements. To balance the loss of different types of push errors, we give different weight to different types of push errors, which is similar with cost-sensitive learning \cite{elkan2001foundations,zhou2006trainging,lomax2013survey,du2010adapting}.

For different application scenarios, the system focuses on different events according to the need. For example, the small-probability events captures our attention in the minority subsets detection \cite{zieba2015counterterrorism,ando2006information}, while someone prefers the event with high probability in support vector machine (SVM) \cite{vapnik1982estimation}. For applications in wireless mobile Internet, in stage of expanding users, recommending the desired data accurately to attract more users is more important. However, in mature stage, advertising is focused on, since the applications expect to earn more money by advertising though it degrades the user experience. This paper will mainly discuss these two cases.

There are also some new findings when the message importance is taken into account \cite{Ivanchev2016information,Kawanaka2017information,monks2013machine,li2015leveraging,masnick1967on}, and message importance can be used to characterize the concern degree of events. Message importance measure (MIM) was proposed to characterize the message importance of events which can be described by discrete random variable where people pay most attention to the small-probability ones, and it highlights the importance of minority subsets \cite{fan2016message}. In fact, it is an extension of Shannon entropy \cite{shannon2001mathematical,verdu1998fifty} and R{\'e}nyi entropy \cite{renyi1961measures} from the perspective of Fadeev's postulates \cite{zum1957,renyi1961measures}. That is, the first three postulates are satisfied for all of them, and MIM weakens the fourth postulate on the foundation of R{\'e}nyi entropy. Moreover, the logarithmic form and polynomial form are adopted in Shannon entropy and R{\'e}nyi entropy, respectively, while MIM uses the exponential form. \cite{she2017focusing} showed that MIM focuses on the a specific event by choosing corresponding importance coefficient. In fact, MIM has a wide range of applications in big data, such as compressed storage and communication \cite{liu2017non} and mobile edge computing \cite{she2018recognizing}. 

In fact, the superior recommendation sequence should like those generated by users, which means that the data recommendation mechanism agrees with user behavior. To this end, the probability of observing the recommendation sequence with the utilization frequency of the user data should be as close to one as possible in the statistical sense. According to \cite{Elements}, that means the relative entropy between the distribution of recommendation sequence and that of user data should be minimized. We assume the recommendation model pursues best user performance with a certain revenue guarantee in this paper.

We first find a particular recommendation distribution on the best effort in maximizing the probability of observing the recommendation sequence with the utilization frequency of the user data when the expected revenue is provided. Then, its main properties, such as monotonicity and geometrical characteristic, are fully discussed. This optimal recommendation system can be regarded as an information-filtering system, and the importance coefficient determines what events the system prefers to recommend.
The results also show that excessively low expectation of revenue can not constrain recommendation distribution and exorbitant expectation of revenue makes the recommendation system impossible to design. The constraints on the recommendation distribution is true if the minimum average revenue is neither too small nor too large, and there is a tradeoff between the recommendation accuracy and the expected revenue.

It is also noted that the form of this optimal recommendation distribution is the same as that of MIM when the minimum average revenue is neither too small nor too large. The optimal recommendation distribution is determined by the proportion of recommendation value of the corresponding event in total  recommendation value, where the recommendation value is a special weight factor. The recommendation value can be seen as a measure of message importance, since it satisfies the postulates of the message importance. Due to the same form with MIM, the optimal recommendation probability can be given by the normalized message importance measure when MIM is used to characterize the concern of system. Furthermore, when importance coefficient is positive, the small-probability events will be taken more attention. That is magnifying the importance index of small-probability events and lessening that of high-probability events. Therefore, we confirms the rationality of MIM from another perspective in this paper for characterizing the physical meaning of MIM by using data ecommendation system rather than based on information theory. Besides, we expand MIM to the general case whatever the probability of systems interested events is. Since the importance coefficient determines what events set systems are interested in, we can switch to different application scenarios by means of it.
That is, advertising systems are discussed if importance coefficient is positive, while noncommercial systems are adopted if importance coefficient is negative.
Compared with previous works about MIM \cite{fan2016message,she2017focusing,liu2017non}, most properties of optimal recommendation distribution is the same, but a clear definition of the desired event set can be given in this paper. The relationship between utility distribution and MIM was preliminarily discussed in \cite{liu2018switch}.

The main contribution of this paper can be summarized as follows. $(1)$ We put forward an optimal recommendation distribution that makes the recommendation mechanism agrees with user behavior with a certain revenue guarantee, which can improve the design of recommendation strategy. $(2)$ We illuminate that this optimal recommendation distribution is normalized message importance, when we use MIM to characterize the concern of systems, which presents a new physical explanation of MIM from data recommendation perspective. $(3)$ We expand MIM to the general case, and we also discuss the importance coefficient selection as well as its relationship with what events systems focus on.

The rest of this paper is organized as follows. The setup of optimal recommendation is introduced in Section \ref{sec:setup}, including the system model and the discussion of constraints. In Section \ref{sec: The optimal recommendation Distribution}, we solve the problem of optimal recommendation in our system model, and give complete solutions. Section \ref{sec: property} investigates the properties of this optimal recommendation distribution. The geometric interpretation is also discussed in this part. Then, we discuss the relationship between this optimal recommendation distribution and MIM in Section \ref{sec:relationship}. It is noted that recommendation distribution can be seen as normalized message importance in this case. The numerical results is shown and discussed to certificate our theoretical results in Section \ref{sec:numerical}. Section \ref{sec:conclusion} concludes the paper. In addition, the main notations in this paper are listed in Table. \ref{tab:notation}.

\begin{table*}[htbp]
\centering
    \caption{Notations.}\label{tab:notation}
\begin{tabular}{|c|l|}
\bottomrule
Notation & Description \\
\hline
$S$  & The set of all the data\\
\hline
$N$  &  The number of data classes\\
\hline
$S_i$ & The set of data belonging to the $i$-th class\\
& $S_i  \cap S_j =\emptyset$ for $i \ne j$, and $S=S_1 \cup S_2 \cup... S_{N-1} \cup S_{N}$  \\
\hline
$Q=\{q_1,q_2,...,q_N\}$  &   Raw distribution: the probability distribution of information source\\
\hline
$P=\{r_1,r_2,...,r_N\}$  &   Recommendation distribution: the probability distribution of recommended data\\
\hline
$U=\{u_1,u_2,...,u_N\}$  &   Utility distribution: the probability distribution of user's preferred data\\
\hline
$D(P\parallel U)$ &  The relative entropy or Kullback-Leibler (KL) distance between $P$ and $U$\\
\hline
$C_p$  &  The cost of a single data push\\
\hline
$R_p$  &  The earning when the pushed data is liked by the user\\
\hline
$C_n$  &  The cost when the pushed data is not liked by the user\\
\hline
$R_{ad}$  &  The advertising revenue when the pushed data is not liked by the user\\
\hline
$C_m$  &  The cost missing to push a piece of user's desired data\\
\hline
$\beta$  & The target revenue of a single data push \\
\hline
$\varpi$  & The importance coefficient\\
\hline
$\alpha$  & $\alpha=\frac{\beta +C_p-R_p}{ R_{ad}-R_p-C_n-C_m}$\\
\hline
$\gamma_u$  & $\gamma_u=\sum_{i=1}^N u_i^2$\\
\hline
$g(\varpi,V)$  & $g(\varpi,P)=\frac{\sum_{i=1}^N v_i^2 e^{\varpi (1-v_i)}}{\sum_{i=1}^N v_i e^{\varpi (1-v_i)}}$\\
\hline
$f(\varpi,x, U)$  & $f(\varpi,x, U)  = \frac{u_x e^{\varpi(1-u_x)}}{\sum\nolimits_{i=1}^N {u_i e^{\varpi (1-u_i)}}}$\\
\hline
$x$  & $x=1,2,\cdots,N$ is the index of classes\\
\hline


\toprule
\end{tabular}
\end{table*}

\section{System Model}\label{sec:setup}
We consider a recommendation system with $N$ classes of data, as shown in Figure \ref{fig:model}. In fact, data is often stored based on its categories for the convenience of indexing. For example, the news website usually classifies news into the following categories: politics, entertainment, business, scientific, sports, and so on.
    At each time instant, the information source generates a piece of data, which belongs to a certain class with probability distribution $Q$ and would be similar with another probability distribution $U$.
In general, the generated data sequence does not match the preference of the user.
    To optimize the information transmission process, therefore, a recommendation unit is used to determine whether the generated data should be pushed to the user with some deliberately designed probability distribution $U$.
One the one hand, the recommendation unit can make predictions of the user's needs and push some data to the user before he actually starts the retrieval process.
    In doing so, the transmission delay can be largely reduced, especially when the data amount is large.
    On the other hand, the recommending unit enables non-expert users to search and to access their desired data much easier.
Furthermore, we can profit more by pushing some advertisements to the user.
\begin{figure}
  \centerline{\includegraphics[width=11.0cm]{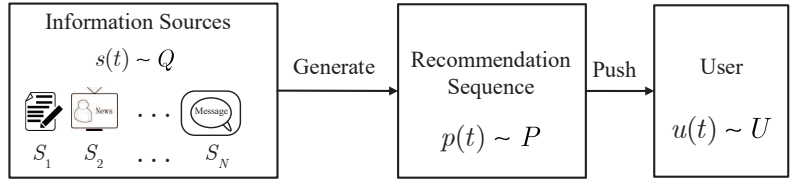}}
  \caption{System model.}\label{fig:model}
\end{figure}
\subsection{Data Model}\label{sec:data model}
We refer to the empirical probability mass function of the class indexes over whole the data set as the \textit{raw distribution} and denote it as $Q=\{q_1,q_2,...,q_N\}$.
      We refer to the probability mass function of users' preference over the classes as the \textit{utility distribution} and denote it as  $U=\{u_1,u_2,...,u_N\}$.
That is, for each piece of data, it belongs to class $i$ with probability $q_i$ and would be preferred by the user with probability $u_i$.
    To fit the preference of the user under some target revenue constraint, the system will make random recommendations according to a \textit{recommendation distribution} $P=\{p_1,p_2,...,p_N\}$.

We assume that each piece of data belongs to one and only one of the $N$ sets.
    That is, $S_i  \cap S_j =\emptyset$ for  $\forall i \ne j$, where $S_i$ is the set of data belonging to the $i$-th class.
Thus, the whole data set would be $S=S_1 \cup S_2 \cup... S_{N-1} \cup S_{N}$.
    Thus, the raw distribution can be expressed as  $q_i= \text{Pr}\{d \in S_i\}=\text{crad}(S_i)/ \text{crad}(S)$.
In addition, the utility distribution $U$ can be obtained by studying the data-using behavior of a specific group of users and thus is assumed to be known in prior in this paper.

For traditional data push, we usually expect to make $|u(t)-s(t)|$ smaller than a given value \cite{bhide2002adaptive}. Different from them, we do not consider this problem based on content. 
As an alternative, our goal is to find the optimal recommendation distribution $P$ so that the recommended data would fit the preference of the user as much as possible.
    To be specific, each recommended sequence of data should resemble the desired data sequence of the user in the statistical sense.
For a sequence of user's favorite data,  let $u^n$ be the corresponding class indexes.
    As $n$ goes to infinity, it is clear that $u^n \in T(U)$ with probability one, where $T(U)$ is the typical set under distribution $U$.
That is, $\Pr\{\frac{1}{n}\log \Pr(u^n)+H(U)=0\}=1$, where $ \Pr(u^n)$ is the occurring probability of $u^n$ and $H(U)$ is the joint entropy of $U$~\cite{Elements}.
    Since the class-index sequence $r^n$ of recommended data is actually generated with distribution $P$, the probability that $r^n$ falls in the typical set  $T(U)$ of distribution $U$ would be $\Pr\{r^n\in T(U)\}\buildrel\textstyle.\over=2^{-nD(P\parallel U)}$, where $D(P\parallel U)$ is the relative entropy between $P$ and $U$ \cite{Elements}.
It is clear that the optimal $P$ would maximizes probability $\Pr\{r^n\in T(U)\}$, which is equivalent to minimizing the relative entropy
    \begin{equation}
        D(P\parallel U).
    \end{equation}
In particular, our desired recommendation distribution $P$ is not exactly the same as the utility distribution of the user, because we also would like to intentionally push some advertisements to the user to increase our profit.

\subsection{Revenue Model}
We assume that the user divides the whole data set into two parts, i.e., the desired ones and the unwanted ones (e.g., advertisements).
    At the recommendation unit, the  data  can also be classified into two types according to whether it is recommended to the user.
    Different push types may strikingly lead to different results. For example, the cost of omitting a data which should be pushed may be much smaller than that of erroneously pushing invalid data to a user. The user experience will be enhanced if the data needed by users is pushed to them. Pushing some not need content to users, such as some advertisement, may seriously impact the user experience, but it still can bring in advertising revenue for the content delivery enterprise.
Using a similar revenue model as that in cost-sensitive learning \cite{elkan2001foundations,zhou2006trainging,lomax2013survey,du2010adapting}, we shall evaluate the revenue of the recommendation system as follows.
    \begin{itemize}
      \item The cost of making a recommendation is $C_p$;
      \item The revenue of a recommendation when the pushed data is liked by the user is $R_p$;
      \item The cost of a recommendation when the pushed data is not liked is $C_n$;
      \item The revenue of a recommendation when the pushed data is is not liked (but can serve as an advertisement) is $R_{ad}$;
      \item The cost of missing to recommend a piece of desired data of the user is $C_m$;
    \end{itemize}
Therefore, the revenue of recommending a piece of data belonging to class $i$ can be summarized in the Table \ref{tab:result1}.

\begin{table*}[htbp]
\centering
    \caption{The revenue matrix.}\label{tab:result1}
\begin{tabular}{|c|c|c|}
\bottomrule
 \diagbox{Preference}{Action} &  Recommend & Not recomend\\
\hline
Desired &  $R_p-C_p$ &  $C_m$\\
\hline
Unwanted  & $R_{ad}-C_p-C_n$ & 0\\
\toprule
\end{tabular}
\end{table*}

Moreover, the corresponding matrix of occurring probability is given by Table \ref{tab:result2}.
\begin{table*}[htbp]
\centering
    \caption{The matrix of occurring probability.}\label{tab:result2}
\begin{tabular}{|c|c|c|}
\bottomrule
 \diagbox{Preference}{Action} &  Recommend & Not recomend\\
\hline
Desired &  $p_iu_i$ &  $(1-p_i)u_i$\\
\hline
Unwanted  & $p_i(1-u_i)$ & $(1-p_i)(1-u_i)$\\
\toprule
\end{tabular}
\end{table*}

In this paper, we assume $C_p,R_p,C_n,R_{ad},C_m$ are constraints for a given recommendation system for simplicity. The expected system revenue can then be expressed as
\begin{flalign}
    \widebar R(P)=& \sum\limits_{i = 1}^N {\left( (R_p- C_p){p_i}{u_i} +(R_{ad}-C_p-C_n) p_i (1-u_i) -C_m (1-p_i)u_i \right)}  \\
    =& -(R_p + C_n + C_m - R_{ad} )(1-\sum\limits_{i = 1}^N p_iu_i) + R_p -C_p \\
    \label{df:R_exp}
    =& -(R_p + C_n + C_m - R_{ad} )\sum\limits_{i = 1}^N p_i(1-u_i) + R_p -C_p.
\end{flalign}

{\subsection{Problem Formulation}
In this paper, we consider the following three kinds of recommendation systems: 1)the \textit{adverting system} where recommending unwanted advertisements yields higher revenue; 2) the \textit{noncommercial system} where recommending user's desired data brings higher revenue; 3) the \textit{neutral system} where the system revenue is independent of recommendation probability $P$.
    For each given target revenue constraint $\widebar R(P) \ge \beta$ and each kind of system, we shall optimize the recommendation distribution $P$ by minimizing $D(P||U)$.
In particular, the following auxiliary variable is used
\begin{equation}\label{df:alpha}
    \alpha=\frac{\beta +C_p-R_p}{ R_{ad}-R_p-C_n-C_m}.
\end{equation}

\subsubsection{Advertising Systems}
 In an \textit{advertising system}, recommending a piece of user's unwanted data (an advertisement) yields higher revenue.
     Since the revenue of recommending an advertisement is the main source of income in this case, which is larger than other revenue and cost, the advertising system would satisfy the following condition.
    \begin{equation}\label{condt:adv_sys}
        \text{C}1:~~ R_p+C_n+C_m-R_{ad} <0.
    \end{equation}

By combining \eqref{df:R_exp}--\eqref{condt:adv_sys}, it is clear that the constraint $\widebar R(P)\ge \beta$ is equivalent to
\begin{equation}
    \sum\limits_{i = 1}^N { {p_i}{(1-u_i)} } \ge \alpha.
\end{equation}

For advertising systems, therefore, the feasible set of recommendation distribution $P$ can be expressed as
\begin{equation}\label{equ: E_condition1}
E_1=\Big\{P: \sum\nolimits_{i = 1}^N { {p_i}{(1-u_i)} } \ge \alpha | R_{ad}>R_p+C_n+C_m \Big\}.
\end{equation}

For a given target revenue $\beta$, we shall solve the optimal recommendation distribution $P$  of  advertising systems from
\begin{flalign}\label{equ:PQ2 optization}
\mathcal{P}_1: \,\, \arg \mathop {\min }\limits_{P\in E_1}\,\,\, &D(P\parallel U)  \\
\textrm{s.t.}\,\,\,& \sum\limits_{i = 1}^N { {p_i}{(1-u_i)} } \ge \alpha  \tag{\theequation a}\label{equ:PQ2 optization a}\\
& \sum\limits_{i=1}^N {p_i}=1 .\tag{\theequation b}\label{equ:PQ2 optization b}
\end{flalign}

\subsubsection{Noncommercial Systems}
A \textit{noncommercial system} is defined as a recommendation system where the revenue $R_p-C_p$ of recommending a piece of desired data is larger than the sum of revenue $R_{ad}-C_p-C_n$ of recommending an advertisement and the cost of not recommending a piece of desired data $C_m$.
    That is,
    \begin{equation}\label{condt:noncom_sys}
        \text{C}2:~~ R_p+C_n+C_m-R_{ad} >0.
    \end{equation}

Accordingly, the constraint $\widebar R(P)\ge \beta$ is equivalent to
\begin{equation}
    \sum\limits_{i = 1}^N { {p_i}{(1-u_i)} } \le \alpha.
\end{equation}

Therefore, the feasible set of recommendation distribution $P$ for  noncommercial systems can be expressed as
\begin{equation}\label{equ: E_condition2}
E_2=\Big\{P: \sum\nolimits_{i = 1}^N { {p_i}{(1-u_i)} } \le \alpha | R_{ad}<R_p+C_n+C_m \Big\}.
\end{equation}

Afterwards, we can solve the optimal recommendation distribution $P$ through the following optimization problem.
\begin{flalign}\label{equ:PQ3 optization}
\mathcal{P}_2: \,\, \arg \mathop {\min }\limits_{P\in E_2}\,\,\, &D(P\parallel U)  \\
\textrm{s.t.}\,\,\,& \sum\limits_{i = 1}^N { {p_i}{(1-u_i)} } \le \alpha  \tag{\theequation a}\label{equ:PQ3 optization a}\\
& \sum\limits_{i=1}^N {p_i}=1 .\tag{\theequation b}\label{equ:PQ3 optization b}
\end{flalign}

\subsubsection{Neutral Systems}
For the case  $R_p+C_n+C_m-R_{ad} =0$, the corresponding expected system revenue degrades to
\begin{equation}
    \widebar R(P) = R_p-C_p
\end{equation}
and is independent from the recommendation distribution $P$.
    As long as the target revenue satisfies $\beta<R_p-C_p$, the constraint $\widebar R(P)\ge \beta$ can be met by any recommendation distribution.
Therefore, the recommendation distribution can be chosen as $P=U$.

\subsubsection{Discussion of Systems}
Noncommercial systems usually appear in stage of expanding users. In this stage, in order to seize market share, the main tasks are attracting more users and making users have confidence in this recommendation systems by excellent user experience, and therefore the revenue from new users by recommending desired data should be larger than the revenue of recommending an advertisement and the cost of not recommending a piece of desired data. Generally, from a qualitative perspective, the data desired by users with high-probability is recommended with higher probability to make $\widebar R(P)$ larger in this case, since $\widebar R(P)$ decreases with increasing of $\sum\nolimits_{i=1}^N {p_i (1-u_i)}$. In this sense, the high-probability events are more important here.

Advertising systems is usually adopted in mature stage of applications where users are accustomed to the recommendation system and can put up with some advertisements. Here, the applications expect to earn more money by advertising. Thus, the revenue of recommending an advertisement is larger than other revenue and cost in advertising systems. From a qualitative perspective, in order to get more revenue, the data desired by users with small-probability (e.g., advertisements) should be recommended with high probability since $\widebar R(P)$ increases with increasing of $\sum\nolimits_{i=1}^N {p_i (1-u_i)}$. In this sense, the small-probability events are more important here.

\begin{rem}
Since $\Big\{P: \widebar R(P)\ge \beta  \Big\}=E_1 \cup E_2 \cup \Big\{P: \widebar R(P)\ge \beta | R_{ad}=R_p+C_n+C_m \Big\}$, these three kinds of recommendation systems cover all the cases of this problem. 
\end{rem}

\section{Optimal Recommendation Distribution}\label{sec: The optimal recommendation Distribution}
In this part, we shall present the optimal recommendation distribution for both advertising systems and noncommercial systems explicitly.
    We define an auxiliary variable and an auxiliary function as follows.
    \begin{align}
        \label{df:gamma_u}
        \gamma_u=&\sum\nolimits_{i = 1}^N { {u_i^2} }, \\
        \label{df:gwv}
        g(\varpi,V)=&\frac{\sum_{i=1}^N v_i^2 e^{\varpi (1-v_i)}}{\sum_{i=1}^N v_i e^{\varpi (1-v_i)}},
    \end{align}
    where $\varpi\in(-\infty, +\infty)$ is a constant and $V=\{v_1,v_2,\cdots,v_N\}$ is a general probability mass function.
    Actually, we have $\gamma_u = e^{-H_2(U)}$ where $H_2(U)$ is the R{\'{e}}nyi entropy $H_{\alpha}(\cdot)$ when $\alpha=2$ \cite{van2014renyi}.

In particular, we have the following lemma on $g(\varpi,V)$.
\begin{lem}\label{lem: monotonicity}
Function $g(\varpi,V)$ is monotonically decreasing with $\varpi$.
\end{lem}
\begin{proof}
Refer to Appendix \ref{app1}.
\end{proof}

\begin{lem}\label{lem: special case}
 $g(0,V)=\sum\nolimits_{i = 1}^N  {v_i^2}$, $g(-\infty,V)=v_{\max}$, and $g(+\infty,V)=v_{\min}$, where $v_{\max}=\max\{v_1, v_2,\cdots,v_N\}$ and $v_{\min}=\min\{v_1, v_2,\cdots,v_N\}$.
\end{lem}
\begin{proof}
Refer to the Appendix \ref{app2}
\end{proof}

It is clear that we have $g(0,P)=\gamma_p$ and $g(0,U)=\gamma_u$.

\subsection{Optimal Advertising System}
\begin{thm}\label{thm: Optimal P in E1}
For an advertising system with $R_p+C_n+C_m-R_{ad} <0$, the optimal recommendation distribution is the solution of Problem $\mathcal{P}_1$ and is given by
\begin{equation}\label{equ: res of Optimal P in E1}
p^*_x  =\left\{
   \begin{aligned}
   & \quad\quad u_x   &&\text{if}~~\alpha \le 1-\gamma_u  \\
   & \frac{u_x e^{\varpi^*(1-u_x)}}{\sum\nolimits_{i=1}^N {u_i e^{\varpi^* (1-u_i)}}} &&\text{if}~~ 1-\gamma_u <\alpha \le1\\
   & \quad\quad \text{NaN}   &&\text{if}~~1\le \alpha   \\
   \end{aligned}
   \right.
\end{equation}
for $1\le x\le N$, where $\alpha$ is defined in \eqref{df:alpha}, $\gamma_u$ is defined in \eqref{df:gamma_u}, NaN means no solution exists, and $\varpi^*>0$ is the solution to $g(\varpi, U)  = 1-\alpha$.
\end{thm}

\begin{proof}
First, if $\alpha>1$, no solution exists since $\sum\nolimits_{i = 1}^N { {p_i}{(1-u_i)} }$ is always smaller than one and the constraint \eqref{equ:PQ2 optization a} can never be satisfied.

Second, if $\alpha \le 1-\gamma_u$, we have $\sum\nolimits_{i = 1}^N { {p_i}{(1-u_i)} }\ge \sum\nolimits_{i = 1}^N { {u_i}{(1-u_i)} } =1-\gamma_u \ge \alpha$.
    That is, the constraint \eqref{equ:PQ2 optization a} could by satisfied with any distribution $P$.
Thus, the solution to Problem $\mathcal{P}_1$ would be $P=U$, which minimizes $D(P||U)$.

Third, if $1-\gamma_u <\alpha \le1$, we shall solve Problem $\mathcal{P}_1$ based on the following Karush-Kuhn-Tucher (KKT) conditions.
\begin{flalign}\label{equ:PQ4 optization}
{\nabla _P}L(P,\lambda,\mu)={\nabla _P} \left({\sum\limits_{i=1}^N {p_i\ln{{p_i} \over {u_i}}}  + \lambda\left({ \alpha-\sum\limits_{i=1}^N {p_i(1-u_i)}}\right)  + \mu \left( {\sum\limits_{i=1}^N {p_i}-1}\right) }\right)=&0  \\
\lambda (\alpha -\sum\limits_{i = 1}^N { {p_i}{(1-u_i)} })=&0  \tag{\theequation a}\label{equ:PQ4 optization a}\\
\sum\limits_{i=1}^N {p_i}-1=&0 \tag{\theequation b}\label{equ:PQ4 optization b}\\
 \alpha-\sum\limits_{i = 1}^N { {p_i}{(1-u_i)} } \le& 0 \tag{\theequation c}\label{equ:PQ4 optization c}\\
  \lambda \ge&0  \tag{\theequation d}\label{equ:PQ4 optization d}
\end{flalign}
Differentiating ${\nabla _P}L(P,\lambda,\mu)$ with respect to $p^*_x$ and set it to zero, we have
\begin{equation}
 \ln p^*_x + 1 - \ln u_x - \lambda (1 - u_x) + \mu =0,
\end{equation}
  and thus $p^*_x=u_x e^{\lambda(1-u_x)-\mu-1}$.
    Together with constraint \eqref{equ:PQ4 optization b}, we further have $e^{\mu+1}=\sum\nolimits_{i=1}^N {u_ie^{\lambda (1-u_i)}}$ and
\begin{equation}\label{equ:PQ4_res}
p^*_x=\frac{u_x e^{\lambda(1-u_x)}}{\sum\nolimits_{i=1}^N {u_i e^{\lambda (1-u_i)}}},
\end{equation}
where $\lambda\ge 0$ is the solution to $\sum\nolimits_{i = 1}^N { {p^*_i}{(1-u_i)} } =\alpha$, i.e., $g(\lambda, U)=1-\alpha$.
By substituting $\lambda$ with $\varpi^*$, the desired result in \eqref{equ: res of Optimal P in E1} would be obtained.

The condition $1-\gamma_u<\alpha \le 1$ implies  $g(\varpi^*, U)=1-\alpha\le \gamma_u=g(0, U)$.
    Since $g(\varpi, U)$ has been shown to be monotonically decreasing with $\varpi$ in Lemma \ref{lem: monotonicity}, we then have $\varpi^*>0$.
This completes the proof of Theorem \ref{thm: Optimal P in E1}.
\end{proof}

\begin{rem}
    We denote
    \begin{align}\label{df:beta_star}
    \beta_0=& (1-\gamma_u)(R_{ad}-R_p-C_n-C_m)+R_p-C_p, \\
    \label{df:beta_1}
    \beta_{ad}=&-(R_{ad}-R_p-C_n-C_m)u_{\min}+R_{ad}-C_p-C_n-C_m,
    \end{align}  and have
    \begin{itemize}
      \item $\alpha \le 1-\gamma_u$ is equivalent to $\beta\le \beta_0$, which means that the target revenue is low and can be achieved by pushing data according to the preference of the user exactly.
      \item $1-\gamma_u<\alpha \le 1$ is equivalent to $\beta_0<\beta\le\beta_{ad}$, which means that the target revenue can only be achieved when some advertisements are pushed according to probability $p^*_x$.
      \item $1<\alpha$ is equivalent to $\beta>\beta_{ad}$, which means that the target is too high to achieve, even when advertisements are pushed with probability one.
    \end{itemize}
\end{rem}

\subsection{Optimal Noncommercial System}
\begin{thm}\label{thm: Optimal P in E2}
For a noncommercial system with $R_p+C_n+C_m-R_{ad} >0$, the optimal recommendation distribution is the solution of Problem $\mathcal{P}_2$ and is given by
\begin{equation}\label{equ: res of Optimal P in E2}
p^*_x  =\left\{
   \begin{aligned}
   & \quad\quad \text{NaN}  &&\text{if}~~\alpha< 0  \\
   & \frac{u_x e^{\varpi^*(1-u_x)}}{\sum\nolimits_{i=1}^N {u_i e^{\varpi^* (1-u_i)}}} && \text{if}~~0 \le \alpha <1-\gamma_u  \\
   & \quad\quad u_x  && \text{if}~~ 1-\gamma_u \le\alpha
   \end{aligned}
   \right. ,
\end{equation}
for $1\le x\le N$, where $\alpha$ is defined in \eqref{df:alpha}, $\gamma_u$ is defined in \eqref{df:gamma_u}, NaN means no solution exists, and $\varpi^*<0$ is the solution to $g(\varpi, U)  = 1-\alpha$.
\end{thm}

\begin{proof}
First, if $\alpha<0$, no solution exists since $\sum\nolimits_{i = 1}^N { {p_i}{(1-u_i)} }$ is positive and cannot be smaller a negative number, and thus the constraint \eqref{equ:PQ3 optization b} can never be satisfied.

Second, if $\alpha \ge 1-\gamma_u$ and we set $p_x^*=u_x$, we would have $\sum\nolimits_{i = 1}^N { {p^*_i}{(1-u_i)} }=1-\gamma_u \le \alpha$, i.e., constraint \eqref{equ:PQ3 optization a} is satisfied.
    Since setting $p_x^*=u_x$ minimizes $D(P||U)$, $P=U$ should be the solution of Problem 2.

Third, if $0\le \alpha < 1-\gamma_u$, we shall solve Problem $\mathcal{P}_2$ using the following KKT conditions.
\begin{flalign}\label{equ:PQ5 optization}
{\nabla _P}L(P,\lambda,\mu)={\nabla _P} \left({\sum\limits_{i=1}^N {p_i\ln{{p_i} \over {u_i}}}  + \lambda\left({ \sum\limits_{i=1}^N {p_i(1-u_i)}-\alpha }\right)  + \mu \left( {\sum\limits_{i=1}^N {p_i}-1}\right) }\right)=&0  \\
\lambda (\sum\limits_{i = 1}^N { {p_i}{(1-u_i)} }-\alpha)=&0  \tag{\theequation a}\label{equ:PQ5 optization a}\\
\sum\limits_{i=1}^N {p_i}-1=&0 \tag{\theequation b}\label{equ:PQ5 optization b}\\
 \sum\limits_{i = 1}^N { {p_i}{(1-u_i)} } -\alpha \le& 0 \tag{\theequation c}\label{equ:PQ5 optization c}\\
  \lambda \ge&0  \tag{\theequation d}\label{equ:PQ5 optization d}
\end{flalign}

Differentiating ${\nabla _P}L(P,\lambda,\mu)$ with respect to $p^*_x$ and set it to zero, we have
\begin{equation}
 \ln p^*_x + 1 - \ln u_x + \lambda (1 - u_x) + \mu = 0,
\end{equation}
 and$p^*_x=u_x e^{-\lambda(1-u_x)-\mu-1}$.
    Together with constraint \eqref{equ:PQ5 optization b}), we then have $e^{\mu+1}=\sum\nolimits_{i=1}^N {u_ie^{-\lambda (1-u_i)}}$ and
\begin{equation}\label{equ:PQ5_res}
p^*_x=\frac{u_x e^{-\lambda(1-u_x)}}{\sum\nolimits_{i=1}^N {u_i e^{-\lambda (1-u_i)}}},
\end{equation}
where $\lambda>0$ is the solution to $\sum\nolimits_{i = 1}^N { {p^*_i}{(1-u_i)} } =\alpha$.

By denoting $\varpi^*=-\lambda$, \eqref{equ:PQ5_res} turns to be $p^*_x=\frac{u_x e^{\varpi^*(1-u_x)}}{\sum\nolimits_{i=1}^N {u_i e^{\varpi^* (1-u_i)}}}$ which is the desired result in \eqref{equ: res of Optimal P in E2}.

 Moreover, the condition $0\le \alpha < 1-\gamma_u$ implies $g(\varpi^*, U)=1-\alpha\ge \gamma_u=g(0,U)$.
    Since $g(\varpi,U)$ is monotonically decreasing with $\varpi$ (cf. Lemma \ref{lem: monotonicity}), we see that $\varpi^*<0$.
 Thus, Theorem \ref{thm: Optimal P in E2} is proved.
\end{proof}

\begin{rem}
        We denote
    \begin{align}\label{df:beta_0}
        \beta_{no}=-(R_{ad}-R_p-C_n-C_m)u_{\max}+R_{ad}-C_p-C_n-C_m,
    \end{align}  and have the following observations.
    \begin{itemize}
      \item $\alpha < 0$ is equivalent to $\beta\ge \beta_{no}$, which means that the target revenue is high and cannot be achieved by any recommendation distribution.
      \item $0\le\alpha < 1-\gamma_u$ is equivalent to $\beta_0<\beta\le\beta_{no}$, which means that the target revenue is not too high and the information is pushed according to probability $p^*_x$. The user experience is limited by the target revenue in this case.
      \item $1-\gamma_u\le\alpha$ is equivalent to $\beta\le\beta_0$, which can be achieved by pushing data according to the preference of the user exactly.
    \end{itemize}
\end{rem}

\subsection{Short Summary}\label{sec: short summary}
We further denote
\begin{equation}
    \beta_{ne} = R_{p}-C_p,
\end{equation}
 the optimal recommendation distributions for various systems and various target revenues can then be summarized in Table \ref{tab:result3}.
\begin{table*}[htbp]
\centering
    \caption{The optimal recommendation distribution.}\label{tab:result3}
\begin{tabular}{|c|c|c|c|c|}
\bottomrule
        & Case &$\beta$ & $\alpha$& $p^*_x$ \\
\hline
\multirow{3}{*}{Advertising system}  &\textcircled{1}&$\beta \le \beta_0$& $\alpha \le 1-\gamma$&$u_x $\\
\cline{2-5}
& \textcircled{2} &$\beta_0<\beta \le \beta_{ad}$ &  $1-\gamma <\alpha \le 1$ & $ \frac{u_x e^{\varpi(1-u_x)}}{\sum\nolimits_{i=1}^N {u_i e^{\varpi (1-u_i)}}}$\\
\cline{2-5}
  &\textcircled{3}& $ \beta > \beta_{ad}$ &$\alpha >1$& NaN\\
\hline
\multirow{2}{*}{Neutral system}  &\textcircled{4} &$\beta \le \beta_{ne}$  & NaN &$u_x$\\
\cline{2-5}
  &\textcircled{5}& $ \beta > \beta_{ne}$ &NaN & NaN\\
\hline
\multirow{3}{*}{Noncommercial system}  &\textcircled{6}& $\beta_0 <\beta \le \beta_{no} $&$0 \le \alpha < 1-\gamma$ &$\frac{u_x e^{\varpi(1-u_x)}}{\sum\nolimits_{i=1}^N {u_i e^{\varpi (1-u_i)}}} $\\
\cline{2-5}
& \textcircled{7}& $\beta\le\beta_0  $ &  $ \alpha \ge 1-\gamma$ & $ u_x$\\
\cline{2-5}
&\textcircled{8}  & $ \beta > \beta_{no}$ &$\alpha <0$& NaN\\
\toprule
\end{tabular}
\end{table*}

Cases \textcircled{3}, \textcircled{5}, and \textcircled{8} are extreme cases where the target revenue is beyond the reach of the system.
    For cases \textcircled{1}, \textcircled{4}, and \textcircled{7}, the target revenue is low ($\beta<\beta_{no}$ or $\beta<\beta_{ad}$), and thus is easy to achieve.
In particular, the constraints \eqref{equ:PQ2 optization a} and \eqref{equ:PQ3 optization a} are actually inactive.
    Thus, the optimal recommendation distribution would be exactly the same with the utility distribution.

Cases \textcircled{2} and \textcircled{6} are more practical and meaningful due to the appropriate target revenues used.
To further study the property of the optimal recommendation distribution of these two cases, the following function is introduced,
\begin{equation}\label{equ: Pc_definition1}
f(\varpi,x, U)  = \frac{u_x e^{\varpi(1-u_x)}}{\sum\nolimits_{i=1}^N {u_i e^{\varpi (1-u_i)}}},
\end{equation}
where $\varpi \in(-\infty,+\infty)$.

In doing so, the optimal recommendation distribution of cases \textcircled{2} and \textcircled{6} can be jointly expressed as
\begin{equation}
    p_x^* = f(\varpi^*,x, U),
\end{equation}
where $\varpi^*$ is the solution to $g(\varpi, U)=1-\alpha$.

    In particular, $f(\varpi^*,x, U)$ presents the optimal solution of case \textcircled{2} if $\varpi^*>0$ and presents the solution of case \textcircled{6} if $\varpi^*<0$.
Moreover, when $\varpi=0$, we have
\begin{equation}\label{equ: f_zero}
f(0,x, U)=u_x,
\end{equation}
which can be considered as the solution to cases \textcircled{1}, \textcircled{4}, and \textcircled{7}.

\section{Property of Optimal Recommendation Distributions}\label{sec: property}
In this section, we shall investigate how the optimal recommendation distribution diverges with respect to the utility distribution, in various systems and under various target revenue constraints.
    To do so, we first study the property of function $f(\varpi,x, U)$, where $\varpi_x$ and $\widetilde{\varpi}_x\neq0$ are, respectively, the solution to the following equations
    \begin{align}
        u_x &= g(\varpi, U),\\
        \label{df:tild_w_x}
        u_x &= f(\varpi, x, U).
    \end{align}

%

\subsection{Monotonicity of $f(\varpi,x, U)$}\label{sec:monotonicity}

\begin{thm}\label{thm: monotonicity}
$f(\varpi,x, U)$ has the following properties.
\begin{enumerate}
  \item[1)] $f(\varpi,x, U)$ is monotonically increasing with $\varpi$ in $(-\infty, \varpi_x)$;
  \item[2)] $f(\varpi,x, U)$ is monotonically decreasing with $\varpi$ in $(\varpi_x, +\infty)$;
  \item[3)] $\varpi_x$ is decreasing with $u_x$, i.e, $\varpi_y<\varpi_x$ if $u_y>u_x$;
  \item[4)] $\varpi_x<0$ if $u_x>\gamma_u$; \\
            $\varpi_x=0$ if $u_x=\gamma_u$; \\
            $\varpi_x>0$ if $u_x<\gamma_u$.
\end{enumerate}
\end{thm}

\begin{proof}
1) and 2) The derivative of $(\varpi,x, U)$ with respect to $\varpi$ can be expressed as
\begin{flalign}\label{equ: proof monotonicity 1}
\frac{ \partial\left(  f(\varpi,x,U)\right)}{ \partial  \varpi}&= \frac{u_x e^{\varpi(1-u_x)} (1-u_x) \sum_{i=1}^N {u_i e^{\varpi (1-u_i)}}-u_x e^{\varpi(1-u_x)}  \sum_{i=1}^N {u_i e^{\varpi (1-u_i)}}(1-u_i) }{ {\left( \sum_{i=1}^N {u_i e^{\varpi (1-u_i)}} \right)}^2} \\
&=\frac{u_x e^{\varpi(1-u_x)}  \sum_{i=1}^N {u_i (u_i-u_x) e^{\varpi (1-u_i)}}}{ {\left( \sum_{i=1}^N {u_i e^{\varpi (1-u_i)}} \right)}^2} . \tag{\theequation a}\label{equ: proof monotonicity 1 a}
\end{flalign}

Since we have $u_x e^{\varpi(1-u_x)}>0$ and ${\left( \sum\nolimits_{i=1}^N {u_i e^{\varpi (1-u_i)}} \right)}^2>0$,  the sign of the derivative only depends on the term $\sum_{i=1}^N {u_i (u_i-u_x) e^{\varpi (1-u_i)}}$.

Lemma \ref{lem: monotonicity} and Lemma \ref{lem: special case} show that $g(\varpi,U)$ is monotonically decreasing with $\varpi$ and $\varpi_x$ is the unique solution to equation $u_x=g(\varpi,U)$.
    For $\varpi<\varpi_x$, therefore, we then  have $g(\varpi,U)>g(\varpi_x,U)=u_x$, which is equivalent to $\sum\nolimits_{i = 1}^N {{u_i}({u_i} - {u_x}){e^{\varpi (1 - {u_i})}}} >0$.
Thus, we have $\frac{ \partial\left(  f(\varpi,x,U)\right)}{ \partial  \varpi}>0$ and $f(\varpi,x,U)$ is increasing with $\varpi$ if $\varpi<\varpi_x$.
    Likewise, it can be readily proved that $f(\varpi,x,U)$ is decreasing with $\varpi$ if $\varpi>\varpi_x$.

3) Suppose $g(\varpi_x,U)=u_x$,  $g(\varpi_y,U)=u_y$ and $u_x<u_y$, we would have, $g(\varpi_x,U)<g(\varpi_y,U)$. 
        Since $g(\varpi)$ is decreasing with $\varpi$ (cf. Lemma \ref{lem: monotonicity}), we then have $\varpi_x>\varpi_y$, i.e., $\varpi_x$ is decreasing with $u_x$. 

4) First, we have $g(0,U)=\gamma_u$ by the definition of $g(\varpi, U)$ (cf. \eqref{df:gwv}). 
    Using the monotonicity of $g(\varpi, U)$ with respect to $\varpi$ (cf. Lemma \ref{lem: monotonicity}), we have $\varpi_x<0$ if $u_x>\gamma_u$ and $\varpi_x>0$ if $u_x<\gamma_u$.

Thus, the proof of Theorem \ref{thm: monotonicity} is completed. 
\end{proof}

In particular, according to Lemma \ref{lem: special case}, if $u_x$=$u_{\min}$, $\varpi_x$ will approach positive infinity, while $\varpi_x$ will approach negative infinity if $u_x$=$u_{\max}$.
\begin{rem}\label{rem: mono_min_max}
\begin{enumerate}
\item[1)] $f(\varpi,x,U)$ is monotonously decreasing  with $\varpi$ if $u_x=u_{\max}$;
\item[2)] $f(\varpi,x,U)$ is monotonously increasing  with $\varpi$ if $u_x=u_{\min}$.
\end{enumerate}
\end{rem}

\begin{rem}
We denote 
\begin{equation}\label{equ: beta cal}
\beta_x=-(R_{ad}-R_p-C_n-C_m)u_x-C_p-C_n-C_m+R_{ad}, 
\end{equation}
and the relationships between $f(\varpi,x, U)$ and $\beta$ in different systems are shown as follows.
\begin{itemize}
\item For advertising systems, 
\begin{enumerate}
  \item[1)] $f(\varpi,x, U)$ is monotonically decreasing with $\beta$ in $( \beta_0,\beta_{ad})$ if $u_x \ge \gamma_u$;
  \item[2)] $f(\varpi,x, U)$ is monotonically increasing with $\beta$ in $(\beta_0, \beta_x)$ and monotonically decreasing in $(\beta_x,\beta_{ad})$ if $u_x < \gamma_u$.
\end{enumerate}
\item For noncommercial systems, 
\begin{enumerate}
  \item[1)] $f(\varpi,x, U)$ is monotonically decreasing with $\beta$ in $( \beta_0,\beta_{no})$ if $u_x \le \gamma_u$;
  \item[2)] $f(\varpi,x, U)$ is monotonically increasing with $\beta$ in $(\beta_0, \beta_x)$ and monotonically decreasing in $(\beta_x,\beta_{no})$ if $u_x < \gamma_u$.
\end{enumerate}
\end{itemize}
\end{rem}




\subsection{Discussion of Parameter}
Assume that there is unique minimum $p_{\min}$ and unique maximum $p_{\max}$ in utility distribution $U$. Without loss of generality, let $u_1 < u_2 \le  ...\le u_{t} \le \gamma_u< u_{t+1} \le ...\le u_{N-1}<u_N$, and $u_{\min}=u_1$ and $u_{\max}=u_N$. $P^*=(p^*_1,p^*_2,...,p^*_N)$ is used to denote the optimal recommendation distribution. Besides, we only discuss the relationship between the optimal recommendation distribution and the parameters ($\beta$ and $\varpi$) in Cases \textcircled{2} and \textcircled{6} in this part.

In fact, we have the following proposition on $\widetilde{\varpi}_x$. 
\begin{prop}\label{pro: 1}
$\widetilde{\varpi}_x$ has the following properties.
\begin{enumerate}
  \item[1)] $\widetilde{\varpi}_x$ exists when $u_x\ne \gamma_u$, $u_x \ne u_{\max}$ and $u_x \ne u_{\min}$;
  \item[2)] $\widetilde{\varpi}_x<0$ if $u_x>\gamma_u$; \\
            $\widetilde{\varpi}_x>0$ if $u_x<\gamma_u$.
 \item[3)] $\widetilde{\varpi}_x$ is decreasing with $u_x$, i.e, $\widetilde{\varpi}_y<\widetilde{\varpi}_x$ if $u_y>u_x$;
\end{enumerate}
\end{prop}
\begin{proof}
Refer to the Appendix \ref{app3}.
\end{proof}

For convenience, we denote $\widetilde{\varpi}_{1} \buildrel \Delta \over =+\infty$ and $\widetilde{\varpi}_{N} \buildrel \Delta \over =-\infty$.

For advertising systems, the optimal recommendation distribution is given by \eqref{equ: res of Optimal P in E1} and $\varpi^*>0$ (cf. Theorem \ref{thm: Optimal P in E1}). As $\varpi^* \to + \infty$, we have
\begin{flalign}\label{equ:PQ_min1}
   p^*_1=f(+\infty,1,U)=&\mathop {\lim }\limits_{\varpi^*  \to  + \infty } {{{u_{\min }}{e^{\varpi^* (1 - {u_{\min }})}}} \over {\sum_{i = 1}^N {{p_i}{e^{\varpi^* (1 - {u_i})}}} }} \\
   =& \mathop {\lim }\limits_{\varpi^*  \to  + \infty } {{{u_{\min }}{e^{\varpi^* (1 - {u_{\min }})}}} \over {{u_{\min }}{e^{\varpi^* (1 - {u_{\min }})}} + \sum_{{u_i} \ne {u_{\min }}} {{u_i}{e^{\varpi^* (1 - {u_i})}}} }}  \tag{\theequation a}\\
  =&   \mathop {\lim }\limits_{\varpi^*  \to  + \infty } {1 \over {1 + \sum\nolimits_{{u_i} \ne {u_{\min }}} {{{{u_i}} \over {{u_{\min }}}}{e^{\varpi^* ({u_{\min }} - {u_i})}}} }}  \tag{\theequation b}\\
  =&   1 .\tag{\theequation c} \label{equ:PQ_min1 c}
\end{flalign}
Obviously, $p^*_{k}=f(+\infty,k,U)=0$ when $k \ge 2$. Therefore, the utility distribution $P^*$ is $(1,0,0,...,0)$ here.

 Based on Proposition \ref{pro: 1}, if $0<\widetilde{\varpi}_{N_3+1} < \varpi^* < \widetilde{\varpi}_{N_3}$ ($1\le N_3 \le N-1$), we have $p^*_x>u_x$ for $1 \le x \le N_3$ and $p^*_x<u_x$ for $N_3+1 \le x \le N$. The number of optimal recommendation probability which is larger than corresponding utility probability is $N_3$. Let $N_4=N-N_3$. $N_3$ decreases with increasing of $\varpi^*$ (cf. Proposition \ref{pro: 1}). In particular, if $\varpi^*>\widetilde{\varpi}_{2}>0$, only the recommendation probability of the event with smallest utility probability is enlarged and that of other events is reduced, compared to the corresponding utility probability. As the parameter $\varpi^*$ approaches positive infinity, the recommendation distribution will be $(1,0,0,...,0)$. In conclusion, this recommendation system is a special information-filtering system which prefers to push the small-probability events.

For noncommercial systems, we can get the optimal recommendation distribution by \eqref{equ: res of Optimal P in E2} and $\varpi^*<0$ (cf. Theorem \ref{thm: Optimal P in E2}). As $\varpi \to -\infty$, it is noted that
\begin{flalign}\label{equ:PQ_max1}
   p^*_N=f(-\infty,N,U)=&\mathop {\lim }\limits_{\varpi  \to  - \infty } {{{u_{\max }}{e^{\varpi (1 - {u_{\max }})}}} \over {\sum_{i = 1}^N {{u_i}{e^{\varpi (1 - {u_i})}}} }} \\
   =& \mathop {\lim }\limits_{\varpi  \to  - \infty } {{{u_{\max }}{e^{\varpi (1 - {u_{\max }})}}} \over {{u_{\max }}{e^{\varpi (1 - {u_{\max }})}} + \sum\limits_{{u_i} \ne {u_{\max }}} {{u_i}{e^{\varpi (1 - {u_i})}}} }} \tag{\theequation a} \\
  =&   \mathop {\lim }\limits_{\varpi  \to  - \infty } {1 \over {1 + \sum_{{u_i} \ne {u_{\max }}} {{{{u_i}} \over {{u_{\max }}}}{e^{\varpi ({u_{\max }} - {u_i})}}} }}  \tag{\theequation b}\\
  =&   1,   \tag{\theequation c} \label{equ:PQ_max1 c}
\end{flalign}
and $p^*_{k}=f(-\infty,k,U)=0$ when $k \le N-1$. Hence, $P^*=(0,0,..,0,1)$ in this case.

If $\widetilde{\varpi}_{K+1} < \varpi < \widetilde{\varpi}_{K}<0$ ($1\le K\le N-1$), $p^*_x<u_x$ for $1 \le x \le K$ and $p^*_x>u_x$ for $K+1 \le x \le N$. The number of optimal recommendation probability which is larger than corresponding utility probability is $N-K$. Let $N_2=N-K$ and $N_1=K$. It is noted that $N_2$ decreases with decreasing of $\varpi$ (cf. Proposition \ref{pro: 1}). In particular, if $\varpi<\widetilde{\varpi}_{N-1}<0$, only the recommendation probability of the event with largest utility probability is enlarged and that of other events is reduced, compared to the corresponding utility probability. As the parameter $\varpi$ approaches negative infinity, the push distribution will be $(0,0,...,0,1)$. In this case, the high-probability events are prefered to push by this recommendation system.

Let the optimal recommendation distribution be equal to the utility distribution, i.e., $P^*= U$, and we have
\begin{flalign}\label{equ: PQ_equ}
{u_j} &= {{{u_j}{e^{\varpi (1 - {u_j})}}} \over {\sum\nolimits_{i = 1}^N {{u_i}{e^{\varpi (1 - {u_i})}}} }}=f(\varpi,j,U)   \quad , \,\,\, 1\le j \le N.
\end{flalign}
It is noted that $\varpi=0$ is the solution of (\ref{equ: PQ_equ}) according to (\ref{equ: f_zero}). Since $f(\varpi,1,U)$ is monotonously increasing with $\varpi$ (cf. Remark \ref{rem: mono_min_max}), $f(\varpi,1,U) \ne f(0,1,U)=u_1$ when $\varpi \ne 0$. Thus there exists one and only one root for (\ref{equ: PQ_equ}), which is $\varpi=0$, and $P^*=U$ in this case. Here, all the data types are fairly treated.

For convenience, the relationship between parameter, the optimal recommendation distribution and the utility distribution is summarized in Table \ref{tab:result5}.
\begin{table}[tbp]
\centering
    \caption{Optimal recommendation distribution with parameters. $\downarrow$ is used to denote $p^*_x<u_x$ and $\uparrow$ is used to denote $p^*_x>u_x$.}\label{tab:result5}
\begin{tabular}{c|c|c|c|c|c|c}
\toprule [1 pt]
$\beta$ & $\varpi$ &$x=1$ & $2 \le x \le t$ &$t+1 \le x \le N-1$ &$x=N$& $P^*$\\
\hline
$\beta_{no}$ & $-\infty$ &$0$ & $p^*_x<u_x$&$p^*_x<u_x$& 1&(0,0,...,0,1)  \\
\hline
$(\beta_0,\beta_{no})$ & $(-\infty,0)$ & $p^*_1<u_1$&$p^*_x<u_x$&/& $p_N^*>u_N$& \multirow{3}{*}{$(\underbrace{\downarrow,...,\downarrow}_{N_1},\underbrace{\uparrow,...,\uparrow}_{N_2})$}\\

$(\beta_0,\tilde \beta_x)$ & $(-\infty,\widetilde{\varpi}_x)$ &/ &  /&$p^*_x<u_x$& $p_N^*>u_N$&  \\

$(\tilde \beta_x, \beta_{no})$ & $(\widetilde{\varpi}_x,0)$ & /  &/&$p^*_x>u_x$& $p_N^*>u_N$ & \\
\hline
$\beta_0$ &  0 &$p^*_1=u_1$& $p^*_x=u_x $ &$p^*_x=u_x $& $p_N^*=u_N$&$(u_1,u_2,...,u_N)$\\
\hline
$(\beta_0,\tilde \beta_x)$ & $(0,\widetilde{\varpi}_x)$& $p^*_1>u_1$&  $p^*_x>u_x$& /& /&\multirow{3}{*}{$(\underbrace{\uparrow,...,\uparrow}_{N_3},\underbrace{\downarrow,...,\downarrow}_{N_4})$} \\

$(\tilde \beta_x, \beta_{ad})$ & $(\widetilde{\varpi}_x,+\infty)$& $p^*_1>u_1$& $p^*_x<u_x$  &/& /& \\

$(\beta_0,\beta_{ad})$ &$(0,+\infty)$&$p^*_1>u_1$&  /&$p^*_x<u_x$ &$p_N^*<u_N$ & \\
\hline
$\beta_{ad}$ &$+\infty$&$1$ & $p^*_x<u_x$&$p^*_x<u_x$ & 0 &(1,0,...,0,0)\\
\toprule [1 pt]
\end{tabular}
\end{table}

In fact, the recommendation system can be regarded as an information-filtering system which push data based on the preferences of users \cite{parra2014optimal}. The input and output of this information-filtering system can be seen as utility distribution and optimal recommendation distribution, respectively. 
For advertising systems, compared to the input, on the output port, the recommendation probability of data belonging to the set $\{S_i| 1\le i \le K\}$ is amplified, and that belonging to the set $\{S_i| K+1\le i \le N\}$ is minified, where $1 \le K\le N-1$ and $0<\widetilde{\varpi}_{K+1} < \varpi < \widetilde{\varpi}_{K}$.
Since $u_1 < u_2 \le  ...\le  ...\le u_{N-1}<u_N$, the advertising system is a special information-filtering system which prefers to push the small-probability events.
Assume $\widetilde{\varpi}_{K+1} < \varpi < \widetilde{\varpi}_{K}<0$ and $1 \le K\le N-1$. 
For noncommercial systems, the data with higher utility probability, i.e., $\{S_i| K+1\le i \le N\}$, is more likely to be pushed, and the data with smaller utility probability, i.e., $\{S_i| 1\le i \le K\}$, tends to be overlooked.
Since $u_1 < u_2 \le  ...\le  ...\le u_{N-1}<u_N$, the high-probability events are prefered to push by the advertising system.

\begin{rem}\label{rem: varpi event}
The recommendation system can be regarded as an information-filtering system, and parameter $\varpi$ is like a indicator which can reflect the system performance.
 If $\varpi>0$, the system is a advertising system, which prefers to recommend the small-probability events, while the system is a noncommercial system, which is more likely to recommend the high-probability events if $\varpi<0$. In particular, the system pushes data according to the preference of the user exactly if $\varpi=0$.
\end{rem}


\subsection{Geometric Interpretation}
In this part, we will give a geometric interpretation of optimal recommendation distribution by means of probability simplex. Let $\mathcal{P}$ denote all the probability distributions from an alphabet $\{S_1,S_2,...,S_N\}$, and point $U$ denote the utility distribution. Besides, $P$ is notation to denote recommendation distribution and $P^*$ denotes the optimal recommendation distribution. With the help of the method of types \cite{Elements}, we have Figure \ref{fig:sketch} to characterize the generation of optimal recommendation distribution.
\begin{figure}
  \centerline{\includegraphics[width=8.0cm]{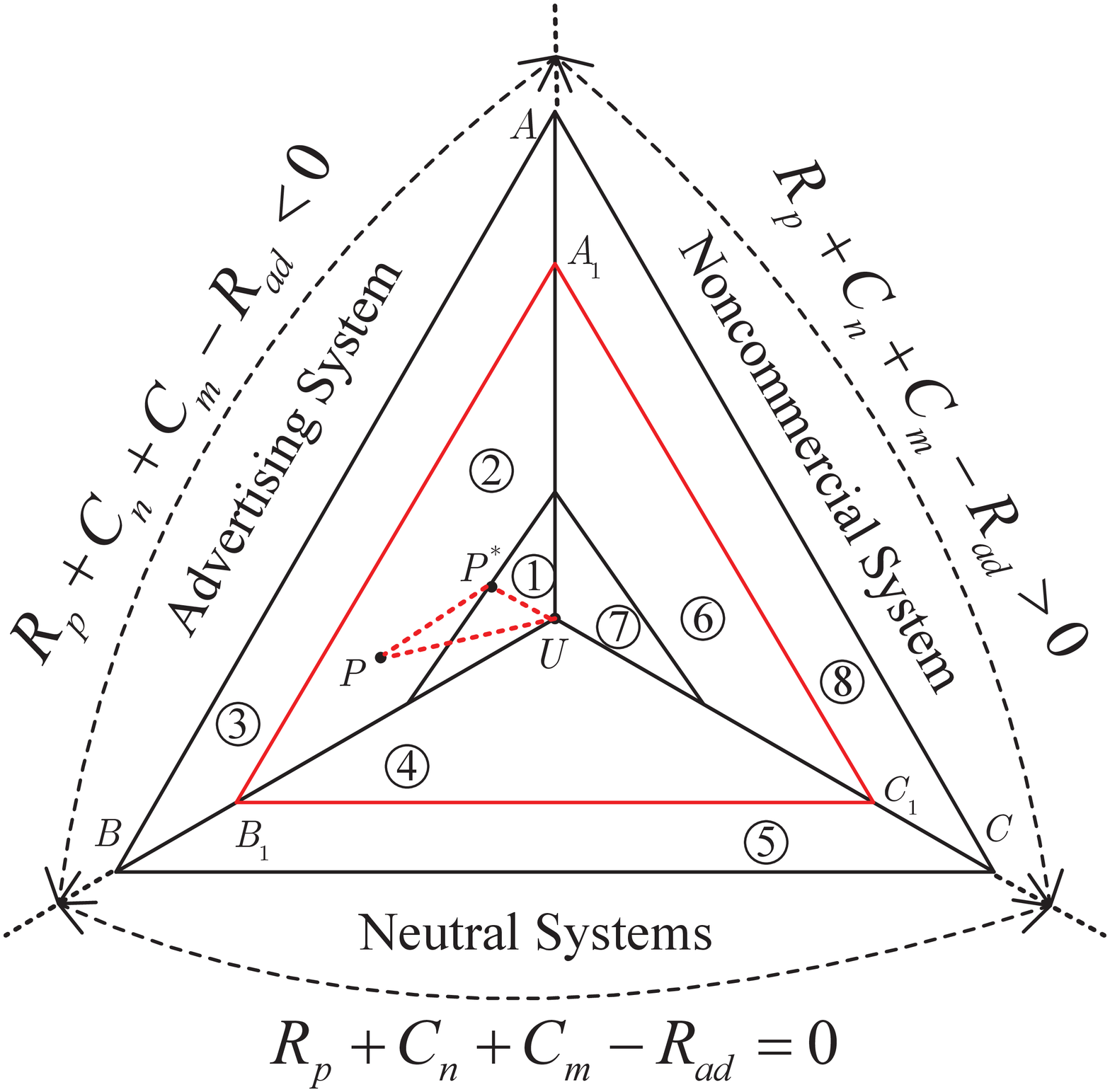}}
  \caption{Probability simplex and optimal recommendation. Region \textcircled{i} denotes Case \textcircled{i} in Table \ref{tab:result2} for $1\le i \le 8$.}\label{fig:sketch}
\end{figure}

In Figure \ref{fig:sketch}, all cases can be grouped into three major categories, i.e., $R_p+C_n+C_m-R_{ad}<0$, $R_p+C_n+C_m-R_{ad}=0$ and $R_p+C_n+C_m-R_{ad}>0$. In fact, these three categories are advertising systems, neutral systems and noncommercial systems, which are denoted by triangles $\Delta_{ABU}$, $\Delta_{BCU}$ and $\Delta_{ACU}$, respectively.
Triangles $\Delta_{A_1 B_1 C_1}$ denotes the region where the average revenue is equal or larger than a given value.
In this paper, our goal is to find the optimal recommendation distribution so that the recommended data would fit the preference of the user as much as possible when the average revenue satisfies the predefined requirement. Since the matching degree between the recommended data and the user's performance is characterized by $2^{-nD(P \parallel U)}$ (cf. Section \ref{sec:data model}), $P^*$ should be a recommendation distribution on the border of or in the triangles $\Delta_{A_1 B_1 C_1}$, which is closest to utility distribution $U$ in relative entropy.
Therefore, there is no solution if recommendation distribution $P$ falls within region \textcircled{3}, \textcircled{5} and \textcircled{8}. 

For advertising systems, $\mathcal{P}$ can be divided into region \textcircled{1}\textcircled{2}\textcircled{3}. Obviously, $P^*=U$ for region $\textcircled{1}$ since $U$ falls within region \textcircled{1}. $P^*$ in \textcircled{2} is the recommendation distribution closest to $U$ in relative entropy. It is noted that $2^{-nD(P^* \parallel U)}>2^{-nD(P \parallel U)}$ if $P \ne P^*$ and $P \in \textcircled{2}$. There is no solution in region \textcircled{3} since $P  \notin \Delta_{A_1 B_1 C_1} $. There is a similar situation for noncommercial systems, which is composed of region \textcircled{6}, \textcircled{7} and \textcircled{8}. There are only two regions for neutral systems.In this case, there is no solution when $P \in \textcircled{5}$, and $P^*=U$ when $P \in \textcircled{4}$.

Furthermore, Triangles $\Delta_{A B U}$ characterizes the set $E_1$ and triangles $\Delta_{A C U}$ characterizes the set $E_2$.

\section{Relationship between optimal recommendation Strategy and MIM} \label{sec:relationship}
\subsection{Normalized Recommendation Value}
In this section, we shall focus on the relationship between optimal recommendation distribution and MIM in Cases \textcircled{2} and \textcircled{6}. In fact, the optimal recommendation distribution in other cases is invariable and it makes little sense to discuss the relationship between them.

Within a period of time, the parameters $C_p,R_p,C_n,R_{ad},C_m$ can be seen as invariants of the recommendation system for a given system, and they do not change with the users. The strategy recommendation is determined by the personalized user concern and the expected revenue in this paper. The former is characterized by the utility distirbution, and the latter is denoted by $\beta$. Based on the discussion in Section \ref{sec: short summary}, there is one-to-one mapping between the parameter $\varpi$ and the minimum average revenue $\beta$. Thus, once the recommendation environment is determined, we only need to select the proper parameter $\varpi^*$ to get the optimal recommendation distribution based on the utility distribution, and here $\varpi^*$ is chosen to satisfy the expectation of the average revenue.

The form of optimal recommendation distribution in (\ref{equ: Pc_definition1}) suggests that we allocate recommendation proportion by weight factor $u_i e^{\varpi^*(1-u_i)}$ where $U=\{u_1,u_2,...,u_N\}$ is the utility distribution. In a sense, we take the weight factor $u_i e^{\varpi^*(1-u_i)}$ as the recommendation value of data belonging to the $i$-th class, and we generate optimal recommendation distribution based on the recommendation value. It is noted that the optimal recommendation distribution is the normalized recommendation value. Furthermore, the total recommendation value of this user is $\sum\nolimits_{i=1}^N {u_i e^{\varpi (1-u_i)}}$.

In fact, the recommendation values are the subjective view, and they are not objective quantity in nature. They show the relative level of recommendation tendency. If all the recommendation values multiply by the same constant, the suggested recommendation distribution will not change.

\subsection{Optimal Recommendation and MIM}
Generally speaking, there must be a reason of the fact thar the recommendation strategy prefers specific data. According to the discussion in Remark \ref{rem: varpi event}, advertising systems prefer the small-probability events (i.e., advertisement), The system push as many small-probability events as possible with the increasing of the minimum average revenuen since pushing small-probability events is the main source of income for this system. Noncommercial systems prefer the high-probability events, and they push as many high-probability events as possible with the increasing of the minimum average revenue since recommending a piece of desired data is the main source of income for this system.
The preference of the recommendation system means that the system has its own evaluation of the degree of importance for different data. The system prefers to recommend the important data in order to achieve the recommendation target in our model.
 Thus, the recommendation distribution is determined by the importance of data. That is, the recommendation probability of data belonging to the $i$-th class is the proportion of recommendation value $u_i e^{\varpi^*(1-u_i)}$ in total recommendation value. 
 In this sense, the recommendation value gives intuitive description of message importance. Furthermore, their relative magnitudes are more significant than their absolute magnitude. For a given parameter, the recommendation value is a quantity with respect to probability distribution. Hence, the recommendation value can be seen as a measure of message importance from this viewpoint, which agrees with the main general principles for definition of importance in previous literatures  \cite{Ivanchev2016information,Kawanaka2017information,monks2013machine,li2015leveraging,masnick1967on}. Moreover, the total recommendation value characterizes the total message importance.


In fact, the form of this importance measure is extremely similar to that of MIM \cite{fan2016message}. MIM is proposed to measure the message importance in the case where small-probability events contains most valuable information and the parameter $\varpi$ in MIM is called \textit{importance coefficient}. The importance coefficient is always larger than zero in MIM. Furthermore, the MIM highlights the importance of minority subsets and ignores the high-probability event by taking them for grant. As stated in Remark \ref{rem: varpi event}, the small-probability events are highlighted when $\varpi>0$. Therefore, MIM is consistent with the conclusion of this paper.

Besides, as the parameter increased to sufficiently large, the message importance of the event with minimum probability dominates the MIM according to \cite{fan2016message,liu2017non}. It is the same with recommendation value since $\mathop {\lim }\limits_{\varpi  \to  + \infty } {{{u_{\min }}{e^{\varpi (1 - {u_{\min }})}}} \over {\sum_{i = 1}^N {{p_i}{e^{\varpi (1 - {u_i})}}} }}=1 $ (cf. \eqref{equ:PQ_min1}). Furthermore, if $\varpi$ is not a very large positive number, the form of $f(\varpi,x,P)$ ($1<x<N$) is similar with Shannon entropy, which is discussed in \cite{she2019information}. \cite{she2017focusing} discussed the selection of importance coefficient, and it pointed out that the event with probability $u_j$ becomes the principal component in MIM when $\varpi=1/u_j$. Due to the same form, the optimal recommendation distribution also has this conclusion. That is, $f(1/u_j,j,U)$ is the larger than $f(1/u_j,i,U)$ if $i \ne j$, which is shown in Fig. \ref{fig:paper_figure_fw}.

\begin{rem}\label{MIM meaning}
The optimal recommendation probability is normalized message importance by means of MIM. In other word, the normalized message importance can also be seen as a distribution that is closest to utility distribution in relative entropy, when the average revenue of this recommendation system is equal or larger than a predefined value.
\end{rem}
Although the MIM is proposed based on information theory, it can also be given from the viewpoint of data recommendation. In fact, Remark \ref{MIM meaning} characterizes the physical meaning of MIM from data recommendation perspective, which confirms the rationality of MIM in one aspect.

\begin{rem}
We also expands MIM to general cases whatever the probability of users interested event is. The importance coefficient $\varpi$ plays a switch role on the event sets of users attention, that is
\begin{enumerate}
  \item[1)] If $\varpi>0$, the importance index of small-probability events is magnified while that of high-probability events is lessened. In this case, the system prefers to recommend the small-probability events;
  \item[2)] If $\varpi<0$, the importance index of high-probability events is magnified while that of small-probability events is lessened. In this case, the system prefers to recommend the high-probability events;
   \item[3)] If $\varpi=0$, The importance of all events is the same.
\end{enumerate}
\end{rem}

Besides, the value of this parameter $\varpi$ can also give a clear definition of the set which users are interested in. For example, for a given utility distribution $U$ ($u_1<u_2\le...\le u_{N-1}<u_N$), if $0<\widetilde{\varpi}_{K+1} < \varpi < \widetilde{\varpi}_{K}$ ($1 \le K\le N-1$), the sparse events are focused on, and the set of these events is $\{S_i| 1\le i \le K\}$ here. In fact, $\{S_i| 1\le i \le K\}$ give a unambiguous description of the definition of small-probability events in MIM. On the contrary, $\{S_i| K\le i \le N\}$ is the set of the events with high-probability which is highlighted if $\widetilde{\varpi}_{K+1} < \varpi < \widetilde{\varpi}_{K}<0$ ($1 \le K\le N-1$). In particular, if $\varpi> \widetilde{\varpi}_{2}>0$, only the events with minimum probability will be focused on.

\section{Numerical Results} \label{sec:numerical}
In this section, some numerical results are presented to to validate the theoretical founds in this paper.


\subsection{Property of $g(\varpi,V)$}
Fig. \ref{fig:paper_figure_g_w} depicts the function $g(\varpi,V)$ versus parameter $\varpi$. The probability distributions are $V_1=\{0.1,0.2,0.3,0.4\}$, $V_2=\{0.05,0.15,0.3,0.5\}$, $V_3=\{0.13,0.17,0.34,0.36\}$,$V_4=\{0.01,0.11,0.12,0.76\}$ and $V_5=\{0.22,0.25,0.25,0.28\}$. The scaling factor of $\varpi$ is varying from $-100$ to $100$.

Some observations are obtained in Fig. \ref{fig:paper_figure_g_w}. The functions in all the cases are monotonically decreasing with $\varpi$. When $\varpi$ is small enough, i.e., $\varpi=-100$, $g(\varpi,V)$ is close to $\max (V_i)$ for $1\le i \le 5$, While $g(\varpi,V)$ approaches $\min (V_i)$ for $1\le i \le 5$ when $\varpi=100$. In fact, we obatin $\min (V_i)\le g(\varpi,V) \le \max(V_i)$. These results validate Lemma \ref{lem: monotonicity} and Lemma \ref{lem: special case}.
Furthermore, the velocity of change of $g(\varpi,V)$ in region $-20<\varpi<20$ is bigger than that in regions $\varpi<-20$ and $\varpi>20$. The KL distance between these probability distributions and uniform distribution are $0.1536,0.3523,0.1230,0.9153,0.0052$, respectively. Thus, the amplitude of variation of $g(\varpi,V)$ decreases with decreasing of KL distance with uniform distribution.

\begin{figure}
  \centerline{\includegraphics[width=8.0cm]{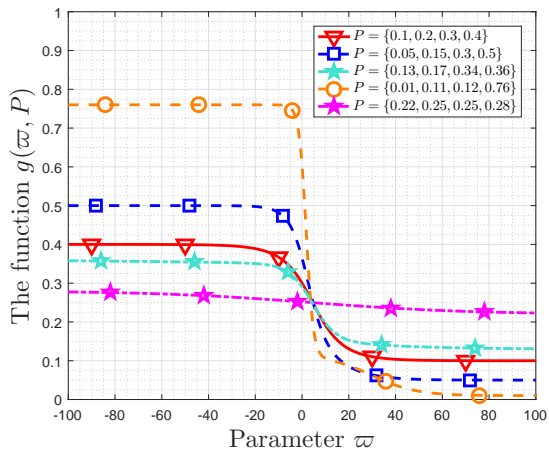}}
  \caption{$g(\varpi,P)$ vs $\varpi$.}\label{fig:paper_figure_g_w}
\end{figure}

\subsection{Optimal Recommendation Distribution}
Then, we focus on proposed optimal recommendation distribution in this paper. The parameters of recommendation system are D1=$\{C_p=4.5,R_p=2,C_n=2,R_{ad}=11,C_m=2\}$ and D2=$\{C_p=1,R_p=9,C_n=2,R_{ad}=3,C_m=2\}$. The minimum average revenue $\beta$ is varying from $0$ to $3$. The utility distributions are U1=$\{0.1,0.2,0.3,0.4\}$ and U2$=\{0.05,0.15,0.3,0.5\}$. Let $U_0=\{0.25,0.25,0.25,0.25\}$. Some results are listed in Table \ref{tab:result6}.
\begin{table*}[htbp]
\centering
    \caption{The auxiliary variables in optimal recommendation distribution.}\label{tab:result6}
\begin{tabular}{c|c|c|c|c|c}
\bottomrule
Case &$R_p+C_n+C_m-R_{ad} $ &$\beta_0$  &  $\beta_{no}$  &$\beta_{ad}$ &$\gamma_u$\\
\hline
D1,U1 &  -5& 1  &  / & 2  &0.3\\
\hline
D2,U1 &  10 &1  &  2 & /&0.3 \\
\hline
D1,U2 &  -5& 0.675  &   /& 2.25&0.365\\
\hline
D2,U2 & 10 & 1.65  &  3 & /&0.365\\
\toprule
\end{tabular}
\end{table*}

 Fig. \ref{fig:paper_figure_Op_dis_1} shows the relationship between optimal recommendation distribution and the minimum average revenue. Some observations can be obtained. 
In fact, the optimal recommendation distribution can be divided into three phases. In phase one, in which the minimum average revenue is small ($\beta<\beta_0$), the optimal recommendation distribution is exactly the same as utility distribution. In phase two, the minimum average revenue is neither too small nor too large ($\beta_0<\beta<\beta_{no}$ or $\beta_0<\beta<\beta_{ad}$). In this case, the optimal recommendation distribution changes with the increasing of minimum average revenue.
In the phase three, in which the minimum average revenue is too large ($\beta>\beta_{no}$ or $\beta>\beta_{ad}$), there is no appropriate recommendation distribution. 

The optimal recommendation probability $p^*_1$ versus the average revenue is depicted in subgraph one of Fig. \ref{fig:paper_figure_Op_dis_1}.
If D1is adopted, which means advertising systems, we obtain $p^*_1$ increases with the increasing of minimum average revenue when $\beta_0<\beta<\beta_{ad}$. $p^*_1$ is larger than $u_1$ in this case. $p^*_1$ approaches one when the average revenue is close to $\beta_{ad}$.
If D2 is adopted, $p^*_1$ decreases with the increasing of minimum average revenue in the region $\beta_0<\beta<\beta_{no}$, and $p^*_1$ is smaller than $u_1$. $p^*_1$ approaches zero as the average revenue increasing to $\beta_{no}$. In addition, $p^*_1=u_1$ when $\beta<\beta_0$.

Subgraph two of Fig. \ref{fig:paper_figure_Op_dis_1} depicts the optimal recommendation probability $p^*_2$ versus the minimum average revenue. If $\beta<\beta_0$, the optimal recommendation probability is equal to the corresponding utility probability. If $\beta>\beta_{no}$ or $\beta>\beta_{ad}$, $p^*_2$ in the four cases will decrease from $u_2$ to zero. In this case, $p^*_2$ increases at the early stage and then decreases for D1,U1 and D1,U2, while $p^*_2$ is monotonously increasing for D2,U1 and D2,U2. Since $u_2<\gamma_u$, $R_p+C_n+C_m-R_{ad}<0$ for D1 and $R_p+C_n+C_m-R_{ad}>0$ for D2, these numerical results agree with the discussion in Section \ref{sec:monotonicity}. Subgraph three is similar to the Subgraph two.

Subgraph four is contrary to the Subgraph one, which shows the optimal recommendation probability $p^*_4$ versus the minimum average revenue. If $\beta<\beta_0$, $p^*_4$ is equal to $u_4$. For D1, i.e., advertising systems, If the minimum average revenue is larger than $\beta_0$ and smaller than $\beta_no$, $p^*_4$ is smaller than $u_4$ and it decreases to zero as $\beta \to \beta_{ad}$. However, $p^*_4$ is larger than $u_4$ and it increases to one as $\beta \to \beta_{no}$ for D2, i.e., noncommercial systems.

\begin{figure}
\begin{minipage}{0.48\linewidth}
  \centerline{\includegraphics[width=8.0cm]{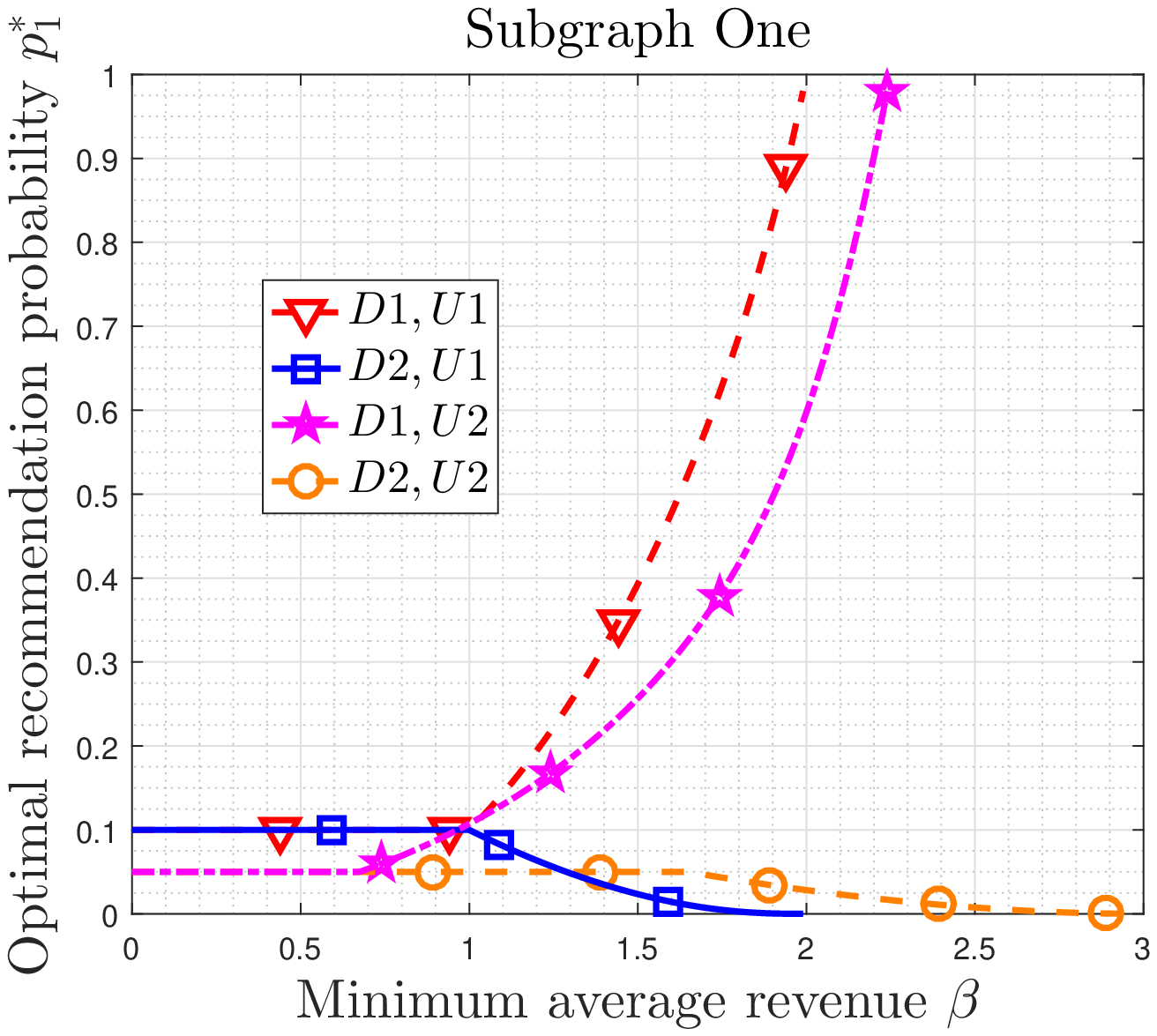}}
\end{minipage}
\hfill
\begin{minipage}{0.48\linewidth}
  \centerline{\includegraphics[width=8.0cm]{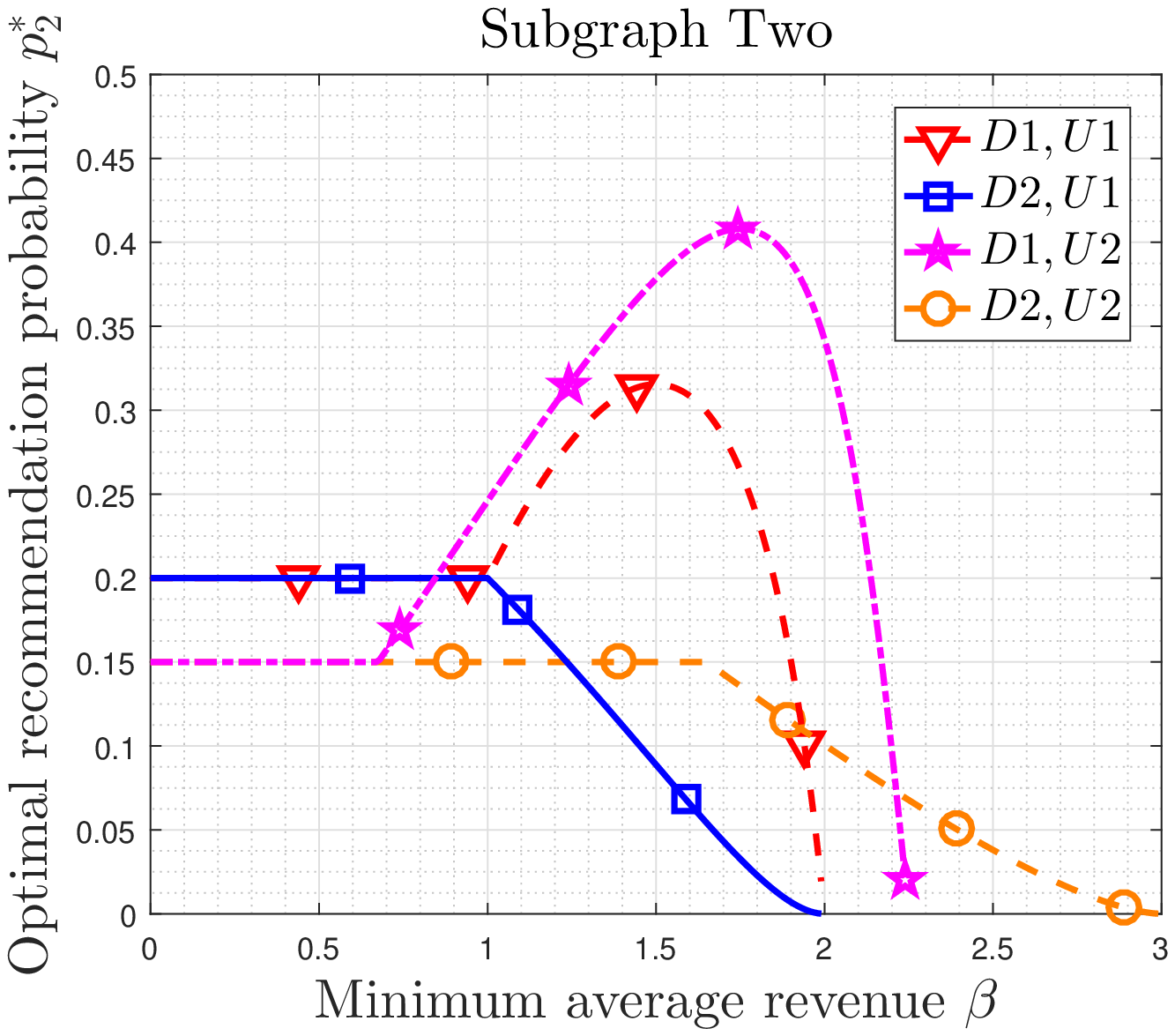}}
\end{minipage}
\begin{minipage}{0.48\linewidth}
  \centerline{\includegraphics[width=8.0cm]{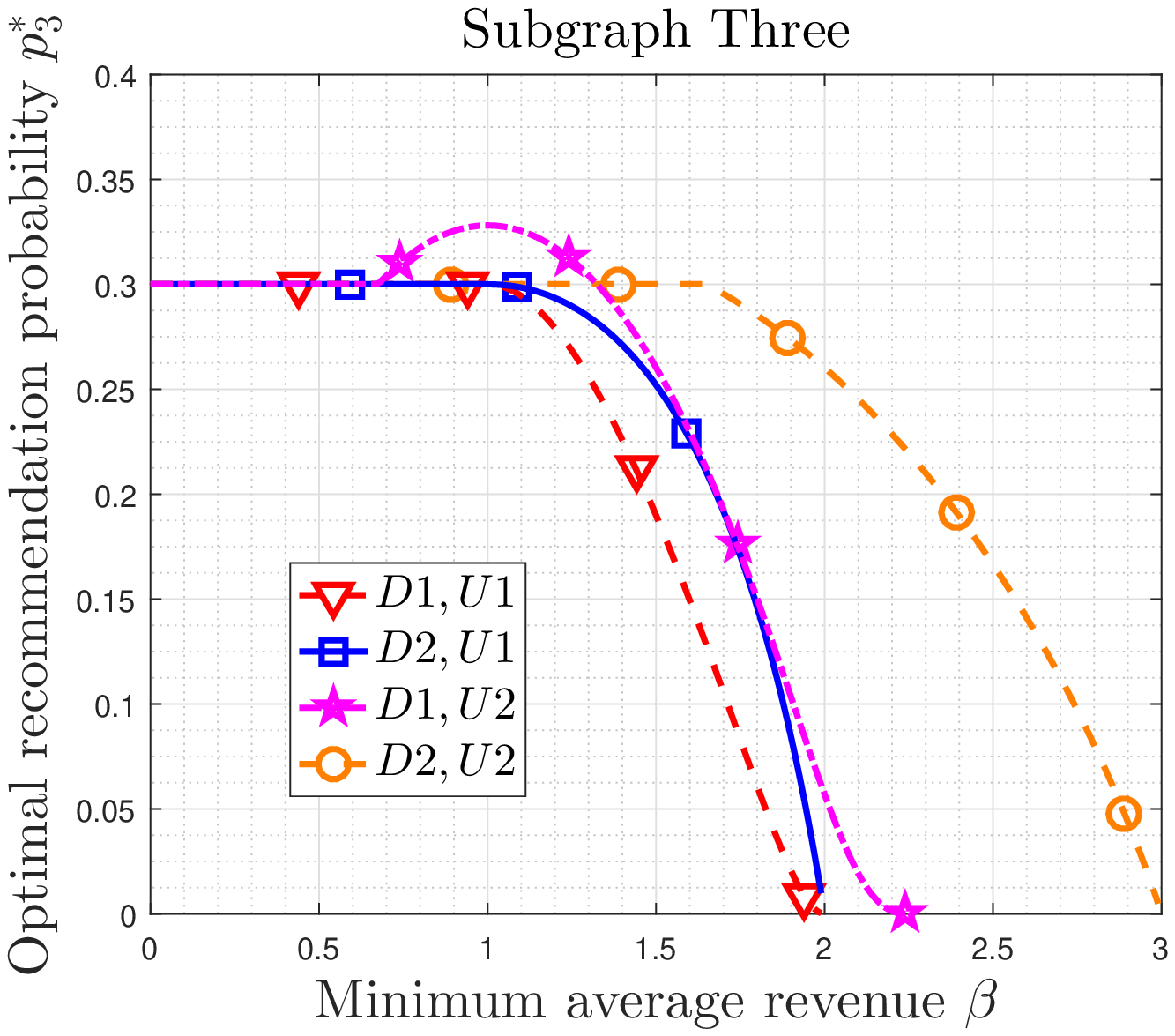}}
\end{minipage}
\hfill
\begin{minipage}{0.48\linewidth}
  \centerline{\includegraphics[width=8.0cm]{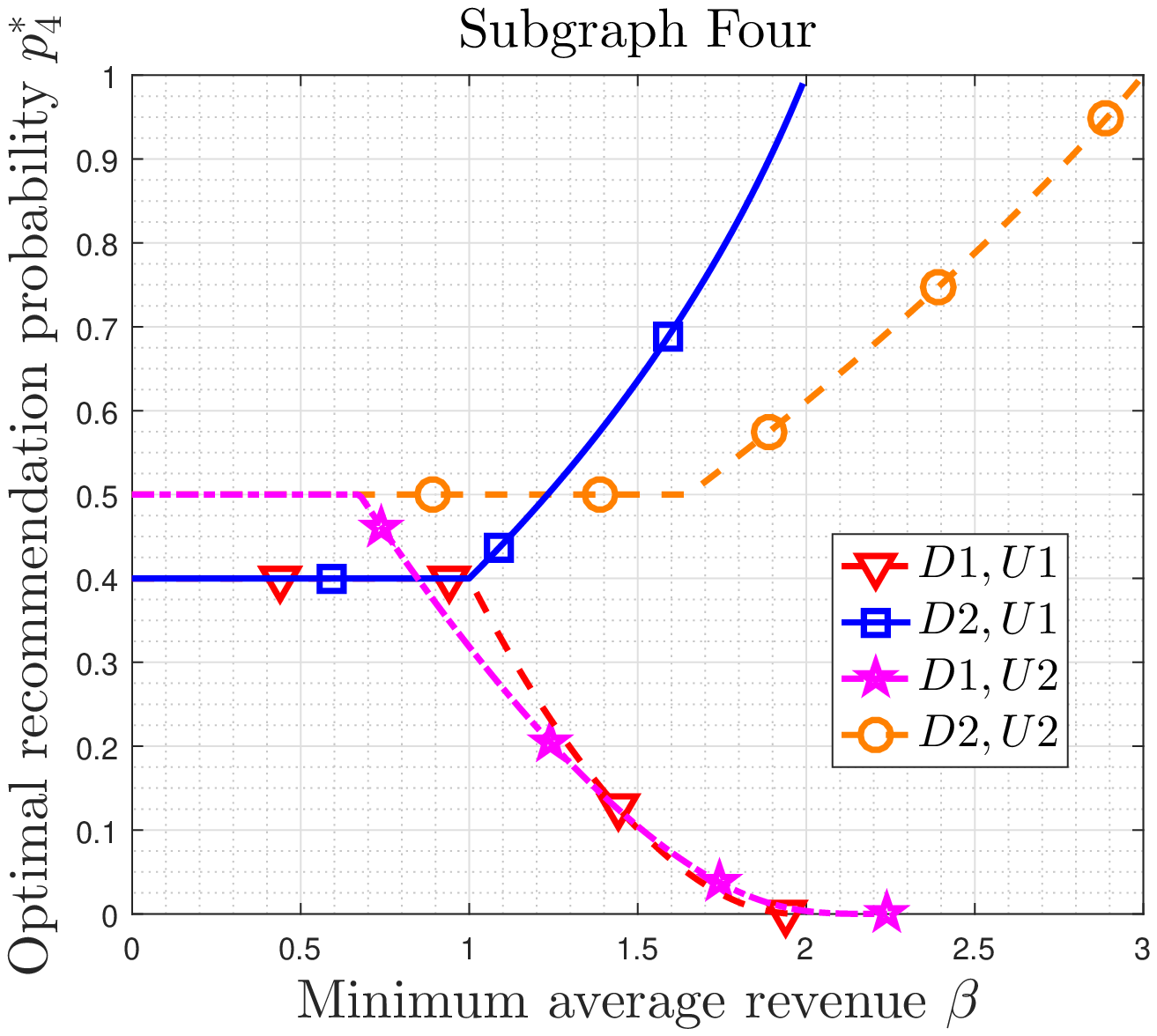}}
\end{minipage}
 \caption{The optimal recommendation distribution vs minimum average revenue. The parameters set $\{C_p,R_p,C_n,R_{ad},C_m\}$ is denoted by D1 and D2, where D1 $=\{4.5,2,2,11,2\}$ and D2 $=\{1,9,2,3,2\}$. The utility distributions are U1$=\{0.1,0.2,0.3,0.4\}$ and U2$=\{0.05,0.15,0.3,0.5\}$.}\label{fig:paper_figure_Op_dis_1}
\end{figure}

Fig. \ref{fig:paper_figure_beta_KL} illustrates minimum KL distance between recommendation distribution and utility distribution versus the minimum average revenue when $\beta_0<\beta<\beta_{no}$ or $\beta_0<\beta<\beta_{ad}$. The constraints on the minimum average revenue is true. The minimum KL distance increases with the increasing of the minimum average revenue. This figure also shows that the minimum average revenue can be gotten for a given minimum KL distance, when the utility distribution and recommendation system parameters are invariant. Since this KL distance represents the accuracy of recommendation strategy, there is a tradeoff between the recommendation accuracy and the minimum average revenue.
user behavior
\begin{figure}
  \centerline{\includegraphics[width=8.0cm]{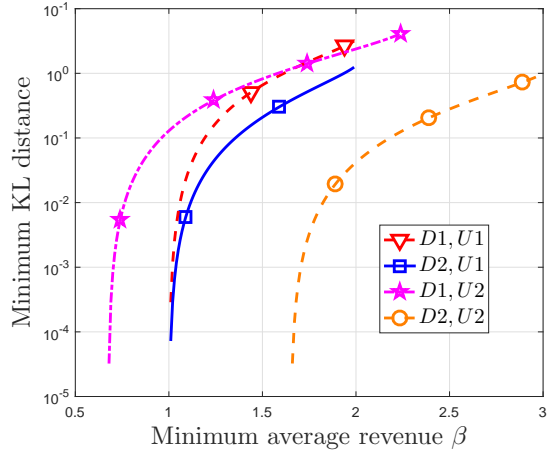}}
  \caption{Minimum KL distance between recommendation distribution and utility distribution vs minimum average revenue. The parameters set $\{C_p,R_p,C_n,R_{ad},C_m\}$ is denoted by D1 and D2, where D1 $=\{4.5,2,2,11,2\}$ and D1 $=\{1,9,2,3,2\}$. The utility distribution is U1$=\{0.1,0.2,0.3,0.4\}$ and U2$=\{0.05,0.15,0.3,0.5\}$.}\label{fig:paper_figure_beta_KL}
\end{figure}
\subsection{Property of $f(\varpi,x,U)$}
Fig. \ref{fig:paper_figure_fw} shows that the importance coefficient $\varpi$ versus the function $f(\varpi,x,U)$. The utility distribution is $\{0.03,0.07,0.12,0.240.25,0.29\}$ and the importance coefficient is varying from $-40$ to $40$. It is easy to check that $u_1<u_2<u_3<\sum\nolimits_{i=1}^6 {u_i^2}<u_4<u_5<u_6$. 
\begin{figure}
  \centerline{\includegraphics[width=16.0cm]{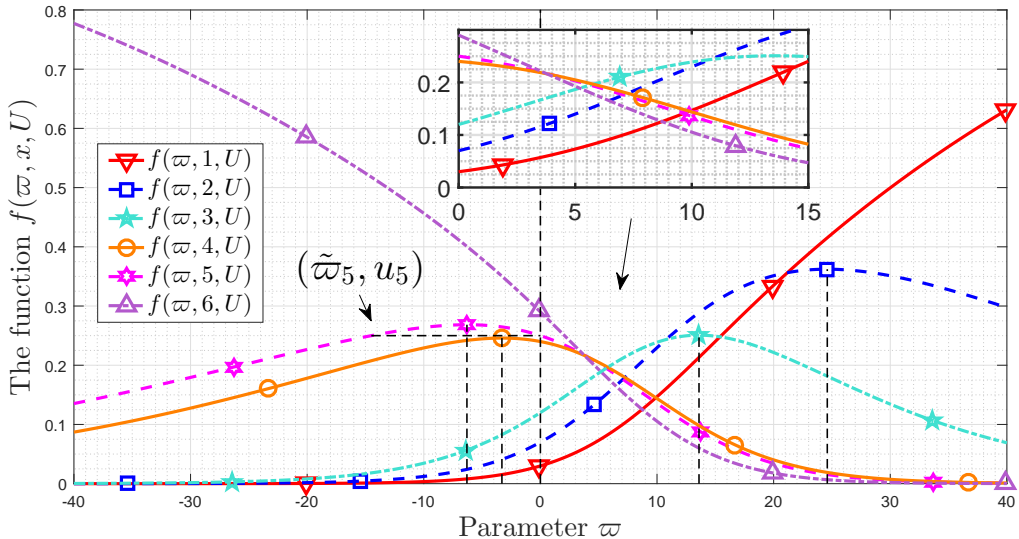}}
  \caption{Function $f(\varpi,x,U)$ vs importance coefficient $\varpi$, when the utility distribution is $U=\{0.03,0.07,0.12,0.24,0.25,0.29\}$.}\label{fig:paper_figure_fw}
\end{figure}

$f(\varpi,1,U)$ is monotonically increasing with increasing of the importance coefficient. $f(\varpi,1,U)$ is close to zero when $\varpi<-30$. 
$f(\varpi,i,U)$ ($i=2,3,4,5$) increases with the increasing of $\varpi$ when $\varpi< \varpi_i$, and then it decreases with the increasing of $\varpi$ when $\varpi> \varpi_i$. They achieve the maximum value when $\varpi=\varpi_i$.
It is also noted that $\varpi_5<\varpi_4<0<\varpi_3<\varpi_2$. 
If $\varpi<-30$, $f(\varpi,i,U)$ ($i=2,3$) is very close to zero, and $f(\varpi,i,U)>0.05$ ($i=2,3$) if $\varpi>30$, which means it changes faster when $\varpi<0$. On the contrary, $f(\varpi,i,U)$ ($i=4,5$) changes faster with $\varpi$ in $(0,+\infty)$ than that in $(-\infty,0)$ since that it is still bigger than $0.5$ when $\varpi=-40$ and it approaches zero when $\varpi>30$. In addition, $f(\varpi,6,U)$ decreases monotonically with the increasing of the importance coefficient, and $f(\varpi,6,U)$ is close to zero when $\varpi>30$.

Some other observations are also obtained. When $\varpi=0$, we obtain $f(\varpi,i,U)=u_i$ ($i=1,2,3,4,5,6$). Without loss of generality, we take $f(\varpi,5,U)$ as an example. There is a non-zero importance coefficient $\tilde \varpi_5$ which makes $f(\varpi,5,U)=u_5$. If $0>\varpi>\widetilde{\varpi}_5$, we obtain $f(\varpi,i,U)>u_i$ for $i=5,6$ and $f(\varpi,i,U)<u_i$ for $i=1,2,3,4$. Compared with the utility distribution, the result of the element in set $\{S_5,S_6\}$ is enlarged, and that in set $\{S_1,S_2,S_3,S_4\}$ is minified. The difference between these two sets is that the utility probability of $S_5$ or $S_6$ is larger than that in $\{S_1,S_2,S_3,S_4\}$.
Besides. it is also noted that $f(\varpi,i,U)>u_i$ for $i=6$ and $f(\varpi,i,U)<u_i$ for $i=1,2,3,4,5$, when $\varpi<\widetilde{\varpi}_5$. Here, only the function output of the element in high-probability set $\{S_6\}$ is larger than the corresponding utility probability.

Furthermore, when $\varpi=1/p_i \in \{33.3333,14.2857,8.3333,4.1667,4,3.4483\}$ for $1\le i\le6$, $f(\varpi,i,U)>f(\varpi,j,U)$ for $j \ne i$.

\section{Conclusion}\label{sec:conclusion}
In this paper, we discussed the optimal data recommendation problem when the recommendation model pursues best user performance with a certain revenue guarantee. Firstly, we illuminated the system model and formulated this problem as optimizations. Then we gave the explicit solution of this problem in different cases, such as advertising systems and noncommercial systems, which can improve the design of data recommendation strategy. In fact, the optimal recommendation distribution is the one that is the closest to the utility distribution in the relative entropy when it satisfies expected revenue. There is a tradeoff between the recommendation accuracy and the expected revenue. In addition, the properties of this optimal recommendation distribution, such as monotonicity and geometric interpretation, were also discussed in this paper. Furthermore, the optimal recommendation system can be regarded as an information-filtering system, and the importance coefficient determines what events the system prefers to recommend. 

We also obtained that the optimal recommendation probability is the proportion of corresponding recommendation value in total recommendation value if the minimum average revenue is neither too small nor too large. In fact, the recommendation value is a special weight factor in determining the optimal recommendation distribution, and it can be regarded as a measure of importance. Since its form and properties are the same as those of MIM, the optimal recommendation probability is exactly the normalized MIM, where MIM is used to characterize the concern of system. 
The parameter in MIM, i.e., importance coefficient, plays a switch role on the event sets of systems attention.
That is, the importance index of high-probability events is enlarged for negative importance coefficient (i.e., noncommercial system), while the importance index of small-probability events is magnified by systems for the positive importance coefficient (i.e., advertising systems).
In particular, only the maximum-probability event or the minimum-probability event are focused on as the importance coefficient approaches to positive infinity or negative infinity, respectively. 
These results give an new physical explanation of MIM from data recommendation perspective, which can validate the rationality of MIM in one aspect.
MIM is extended to the general case whatever the probability of systems interested events is. One can adjust the importance coefficient to focus on desired data type. Compared with previous MIM, the set of desired events can be defined precisely. These results can help us formulate appropriate data recommendation strategy in different scenarios.

In the future ,we may consider its applications in next generation cellular systems \cite{xiong2017group,xiong2014energy}, wireless sensor networks \cite{wang2010delay} and very high speed railway communication systems \cite{zhang2015optimal} by taking the singal transmission mode into accounts.

\appendix
\section{}
\subsection{Proof of Lemma \ref{lem: monotonicity}}\label{app1}
The derivation of $g(\varpi,V)$ is given by
\begin{flalign}\label{equ: proof of lemma 1}
\frac{ \partial  g(\varpi,V)}{ \partial  \varpi}&=\frac{\sum_{i=1}^N v_i^2(1-v_i) e^{\varpi (1-v_i)}\sum_{j=1}^N v_j e^{\varpi (1-v_j)} - \sum_{i=1}^N v_i^2 e^{\varpi (1-v_i)}\sum_{j=1}^N v_j (1-v_j)e^{\varpi (1-v_j)}  }{ {\left( \sum_{i=1}^N v_i e^{\varpi (1-v_i)}  \right)}^2 }  \\
&= \frac{ \sum\nolimits_{i=1}^{N} {\sum\nolimits_{j=1}^{N} {v_i^2 (1-v_i) v_j e^{\varpi(2-v_i-v_j)} }} - \sum\nolimits_{i=1}^{N} {\sum\nolimits_{j=1}^{N} {v_i^2 v_j (1-v_j) e^{\varpi(2-v_i-v_j)} }} }{ {\left( \sum_{i=1}^N v_i e^{\varpi (1-v_i)}  \right)}^2 } \tag{\theequation a} \label{equ: proof of lemma 1 a}\\
&= \frac{ \sum\nolimits_{i=1}^{N} {\sum\nolimits_{j=1}^{N} {v_i^2 v_j (v_j-v_i) e^{\varpi(2-v_i-v_j)} }}  }{ {\left( \sum_{i=1}^N v_i e^{\varpi (1-v_i)}  \right)}^2 } \tag{\theequation b} \label{equ: proof of lemma 1 b}\\
&= \frac{ \sum\limits_{v_i<v_j} {v_i^2 v_j (v_j-v_i) e^{\varpi(2-v_i-v_j)} }  + \sum\limits_{v_i>v_j} {v_i^2 v_j (v_j-v_i) e^{\varpi(2-v_i-v_j)} } }{ {\left( \sum_{i=1}^N v_i e^{\varpi (1-v_i)}  \right)}^2 } \tag{\theequation c} \label{equ: proof of lemma 1 c}\\
&= \frac{ \sum\limits_{v_j<v_i} {v_j^2 v_i (v_i-v_j) e^{\varpi(2-v_j-v_i)} }  + \sum\limits_{v_i> v_j} {v_i^2 v_j (v_j-v_i) e^{\varpi(2-v_i-v_j)} } }{ {\left( \sum_{i=1}^N v_i e^{\varpi (1-v_i)}  \right)}^2 } \tag{\theequation d} \label{equ: proof of lemma 1 d}\\
&= \frac{ \sum\limits_{v_i>v_j} {(v_j-v_i) v_i v_j (v_i-v_j) e^{\varpi(2-v_i-v_j)}   } }{ {\left( \sum_{i=1}^N v_i e^{\varpi (1-v_i)}  \right)}^2 } \tag{\theequation e} \label{equ: proof of lemma 1 e}\\
&= \frac{ \sum\limits_{v_i>v_j} { - {(v_j-v_i)}^2 v_i v_j e^{\varpi(2-v_i-v_j)}   } }{ {\left( \sum_{i=1}^N v_i e^{\varpi (1-v_i)}  \right)}^2 }<0 \tag{\theequation f} \label{equ: proof of lemma 1 f}
\end{flalign}
(\ref{equ: proof of lemma 1 d}) is gotten by exchanging the notation of subscript. (\ref{equ: proof of lemma 1 f}) follows from the fact that $v_i \ge 0$ for $1\le i\le N$ and not all of $v_i$ is zero.

\subsection{Proof of Lemma \ref{lem: special case}}\label{app2}
When $\varpi=0$, we have
\begin{equation}
g(0,V)=\frac{\sum_{i=1}^N v_i^2 }{\sum_{i=1}^N v_i } =\sum\nolimits_{i = 1}^N { {v_i^2} }.
\end{equation}
Let $g(-\infty,V)\buildrel \Delta \over = \mathop {\lim }\limits_{\varpi  \to  - \infty } {g(\varpi,V)}$, and we obtain
\begin{flalign}
  g(-\infty,V)=\mathop {\lim }\limits_{\varpi  \to  - \infty } {{\sum\nolimits_{i = 1}^n {v_i^2{e^{\varpi (1 - {v_i})}}} } \over {\sum\nolimits_{i = 1}^n {{v_i}{e^{\varpi (1 - {v_i})}}} }} &= \mathop {\lim }\limits_{\varpi  \to  - \infty } {{v_{\max }^2{e^{\varpi (1 - {v_{\max }})}} + \sum\limits_{{v_i} \ne {v_{\max }}} {v_i^2{e^{\varpi (1 - {v_i})}}} } \over {{v_{\max }}{e^{\varpi (1 - {v_{\max }})}} + \sum\limits_{{v_i} \ne {v_{\max }}} {{v_i}{e^{\varpi (1 - {v_i})}}} }} \nonumber \\
  &= \mathop {\lim }\limits_{\varpi  \to  - \infty } {{{v_{\max }} + \sum\limits_{{v_i} \ne {v_{\max }}} {{{v_i^2} \over {{v_{\max }}}}{e^{\varpi ({v_{\max }} - {v_i})}}} } \over {1 + \sum\limits_{{v_i} \ne {v_{\max }}} {{{{v_i}} \over {{v_{\max }}}}{e^{\varpi ({v_{\max }} - {v_i})}}} }}=v_{\max}
\end{flalign}
Let $g(+\infty,V)\buildrel \Delta \over = \mathop {\lim }\limits_{\varpi  \to  + \infty } {g(\varpi,V)}$, and it is noted that
\begin{flalign}
g(+\infty,V)=  \mathop {\lim }\limits_{\varpi  \to  + \infty } {{\sum\nolimits_{i = 1}^n {v_i^2{e^{\varpi (1 - {v_i})}}} } \over {\sum\nolimits_{i = 1}^n {{v_i}{e^{\varpi (1 - {v_i})}}} }} &= \mathop {\lim }\limits_{\varpi  \to  + \infty } {{v_{\min }^2{e^{\varpi (1 - {v_{\min }})}} + \sum\limits_{{v_i} \ne {v_{\min }}} {v_i^2{e^{\varpi (1 - {v_i})}}} } \over {{v_{\min }}{e^{\varpi (1 - {v_{\min }})}} + \sum\limits_{{v_i} \ne {v_{\min }}} {{v_i}{e^{\varpi (1 - {v_i})}}} }} \nonumber \\
 &= \mathop {\lim }\limits_{\varpi  \to  + \infty } {{{v_{\min }} + \sum\limits_{{v_i} \ne {v_{\min }}} {{{v_i^2} \over {{v_{\min }}}}{e^{\varpi ({v_{\min }} - {v_i})}}} } \over {1 + \sum\limits_{{v_i} \ne {v_{\min }}} {{{{v_i}} \over {{v_{\min }}}}{e^{\varpi ({v_{\min }} - {v_i})}}} }}=v_{\min}
\end{flalign}
Hence, $v_{\min}< g(\varpi,V) <v_{\max}$ when $\varpi \in (-\infty,+\infty)$ according to Lemma \ref{lem: monotonicity}.

\subsection{Proof of Proposition \ref{pro: 1}}\label{app3}
1) First, if $u_x = u_{\max}$ and $u_x = u_{\min}$, the non-zero solution of equation \eqref{df:tild_w_x} does not exist since $f(\varpi,x,U)$ is monotonously changing (cf. Remark \ref{rem: mono_min_max}). 

Second, if $u_x=\gamma_u$, $u_x \ne u_{\max}$ and $u_x \ne u_{\min}$, no solution exists since $f(0,x,U)>f(\varpi,x,U)$ for $\varpi \ne 0$ according to Theorem \ref{thm: monotonicity}.

Third, if $u_x<\gamma_u$, $u_x \ne u_{\max}$ and $u_x \ne u_{\min}$, according to Theorem \ref{thm: monotonicity}, we have $\varpi_x>0$ and $f(\varpi, x, U)$ is increasing in $(0,\varpi_x)$ while decreasing in $(\varpi_x, +\infty)$, where $\varpi_x$ is the solution to $g(\varpi,U)=1-\alpha$. It is easy to check that $f(+\infty,x,U) -f(0,x,U)=-u_x<0$ (cf. \eqref{equ:PQ_min1}) and $f(\varpi_x,x,U) -f(0,x,U)>0$. According to zero point theorem, the non-zero solution to $f(\varpi, x, U)=u_x$, i.e., $\widetilde{\varpi}_x$ (cf. \eqref{df:tild_w_x}), would be found in $(\varpi_x, +\infty)$.

Fourth, likewise, the non-zero solution $\widetilde{\varpi}_x$ can be found in $(-\infty,\varpi_x)$ if $u_x>\gamma_u$, $u_x \ne u_{\max}$ and $u_x \ne u_{\min}$.

2) and 3) Consider $u_y<u_x<\gamma_u$ first. 
According to Theorem \ref{thm: monotonicity}, we have $\varpi_x>0$ and $f(\varpi, x, U)$ is increasing in $(0,\varpi_x)$ while decreasing in $(\varpi_x, +\infty)$, where $\varpi_x$ is the solution to $g(\varpi,U)=1-\alpha$.
    Since we have $f(0, x, U)=u_x$, the non-zero solution to $f(\varpi, x, U)=u_x$, i.e., $\widetilde{\varpi}_x$ (cf. \eqref{df:tild_w_x}), would only be found in $(\varpi_x, +\infty)$. Likewise, we obtain that $\widetilde{\varpi}_y>\varpi_y>0$. 
    Since $\varpi_y>\varpi_x$ when $u_y<u_x$ (cf. Theorem \ref{thm: monotonicity}), we have $\widetilde{\varpi}_x>\varpi_x>0$ and $\widetilde{\varpi}_y >\varpi_y>\varpi_x>0$.

Furthermore, $f(\widetilde{\varpi}_y,y, U)=u_y$ implies $ \frac{e^{\widetilde{\varpi}_y(1-u_y)}}{\sum\nolimits_{i=1}^N {u_i e^{\widetilde{\varpi}_y (1-u_i)}}}=1$. Hence, we have 
\begin{flalign}\label{equ:proof of pro}
f(\widetilde{\varpi}_y,x, U)&=\frac{e^{\widetilde{\varpi}_y(1-u_x)}}{\sum\nolimits_{i=1}^N {u_i e^{\widetilde{\varpi}_y (1-u_i)}}}\\
&=\frac{e^{\widetilde{\varpi}_y(u_y-u_x)} e^{\widetilde{\varpi}_y(1-u_y)}}{\sum\nolimits_{i=1}^N {u_i e^{\widetilde{\varpi}_y (1-u_i)}}}    \tag{\theequation a} \label{equ:proof of pro a}\\
&<u_x=f(\widetilde{\varpi}_x,x, U),    \tag{\theequation b} \label{equ:proof of pro b}
\end{flalign}
where (\ref{equ:proof of pro b}) follows $ \frac{e^{\widetilde{\varpi}_y(1-u_y)}}{\sum\nolimits_{i=1}^N {u_i e^{\widetilde{\varpi}_y (1-u_i)}}}=1$, $u_x>u_y$ and $\widetilde{\varpi}_y>0$.

Since $f(\varpi,x, U)$ is decreasing with $\varpi$ according to Theorem \ref{thm: monotonicity}, we then have $\widetilde{\varpi}_y>\widetilde{\varpi}_x$.

Second, likewise, we can prove that $\widetilde{\varpi}_x<\widetilde{\varpi}_y<0$ if $u_x>u_y>\gamma_u$.

Third, if $u_y<\gamma_u < u_x$, we shall have $\widetilde{\varpi}_x<0$ and $\widetilde{\varpi}_y>0$. Obviously, $\widetilde{\varpi}_y>\widetilde{\varpi}_x$ in this case.

Thus, the proof of Proposition \ref{pro: 1} is completed. 






\end{document}